\documentclass[11pt]{article}

\usepackage[utf8]{inputenc}

\usepackage{amsmath,amsfonts,amsthm,amssymb,color}
\usepackage[dvipsnames,svgnames,table]{xcolor}
\definecolor{darkgreen}{rgb}{0.0,0,0.9}
\usepackage{mathtools}
\usepackage{authblk}
\usepackage{fullpage}
\usepackage{multirow}
\usepackage{parskip}
\usepackage{comment}
\usepackage{tikz}
\usepackage{bbm}
\usepackage{todonotes}
\presetkeys{todonotes}{inline}{}
\usepackage{dsfont}
\usepackage[sc]{mathpazo}
\usepackage[basic]{complexity}
\usepackage{algorithm2e}
\usepackage[colorlinks=true,citecolor=OliveGreen,linkcolor=BrickRed,urlcolor=BrickRed,pdfstartview=FitH]{hyperref}
\usepackage[capitalize,nameinlink]{cleveref}
\usepackage{tcolorbox}
\usepackage{thm-restate}
\usepackage{natbib}
\usepackage{caption}
\usepackage{booktabs}
\newtcolorbox{wbox}
{
	colback  = white,
}

\SetKwInOut{Input}{Input}
\SetKwInOut{Output}{Output}
\SetKwFunction{Uncross}{\textsc{Uncross}}
\SetKwFunction{MergeUncross}{\textsc{MergeUncross}}
\SetKwFunction{ConnectedComponents}{\textsc{ConnectedComponents}}
\SetKwFunction{PerfectMatching}{\textsc{PerfectMatching}}
\SetKwBlock{InParallel}{in parallel do}{end}
\SetKwFor{ParallelFor}{for}{in parallel do}{end}

\newcommand*{\suppress}[1]{}

\let\poly\relax
\DeclareMathOperator{\poly}{poly}

\makeatletter
\def\thm@space@setup{%
	\thm@preskip= 10pt
	\thm@postskip=\thm@preskip 
}
\makeatother

\makeatletter
\renewcommand{\paragraph}{%
	\@startsection{paragraph}{4}%
	{\z@}{5pt}{-1em}%
	{\normalfont\normalsize\bfseries}%
}
\makeatother

\newtheorem{theorem}{Theorem}
\newtheorem{lemma}{Lemma}

\newtheorem{definition}{Definition}

\newtheorem{proposition}[theorem]{Proposition}

\newtheorem{question}[theorem]{Question}

\theoremstyle{definition}
\newtheorem{remark}[theorem]{Remark}

\newenvironment{fminipage}%
{\begin{Sbox}\begin{minipage}}%
		{\end{minipage}\end{Sbox}\fbox{\TheSbox}}

\def\intersect{\cap}

\def\calL{\mathcal{L}}

\def\calW{\mathcal{W}}
\def\calF{\mathcal{F}}
\def\calM{\mathcal{M}}

\newcommand{\Lc}{\mathcal{L}}

\newcommand*\samethanks[1][\value{footnote}]{\footnotemark[#1]}

\title{Stable Matching: Dealing with Changes in Preferences}

 \author[1]{Rohith Reddy Gangam\thanks{This work was supported in part by NSF grant CCF-2230414.}}
 \author[2]{Tung Mai}
 \author[3]{Nitya Raju\footnote{This work was done while the author was a student at the University of California, Irvine.}}
 \author[1]{Vijay V.~Vazirani\samethanks[1]}

 \affil[1]{University of California, Irvine}
 \affil[2]{Adobe Research}
 \affil[3]{University of Maryland, College Park}

\date{}

\begin{document}
\maketitle
    
\begin{abstract}
We study stable matchings that are robust to preference changes in the two-sided stable matching setting of Gale and Shapley~\cite{GaleS}. Given two instances $A$ and $B$ on the same set of agents, a matching is said to be \emph{robust} if it is stable under both instances. While prior work has considered the case where a {\em single agent} changes preferences between $A$ and $B$, we allow {\em multiple agents on both sides} to update their preferences and ask whether three central properties of stable matchings extend to robust stable matchings: (i) Can a robust stable matching be found in polynomial time? (ii) Does the set of robust stable matchings form a lattice? (iii) Is the fractional robust stable matching polytope integral?

We show that all three properties hold when any number of agents on one side change preferences, as long as at most one agent on the other side does. For the case where two or more agents on both sides change preferences, we construct examples showing that both the lattice structure and polyhedral integrality fail—identifying this setting as a sharp threshold. We also present an XP-time algorithm for the general case, which implies a polynomial-time algorithm when the number of agents with changing preferences is constant. While these results establish the tractability of these regimes, closing the complexity gap in the fully general setting remains an interesting open question.
\end{abstract}

\section{Introduction}
\label{sec.intro}

Matchings under preferences form a fundamental area of research in algorithmic game theory, with the seminal 1962 paper by Gale and Shapley \cite{GaleS} launching a rich line of work on stable matchings. Over the decades, this field has revealed elegant structural, algorithmic, and strategic properties, leading to real-world applications ranging from labor markets to school choice \cite{Book-Online}. Unlike the setting defined in \cite{GaleS}, we note that  in practice, preferences of agents are often not static. They may change as agents gain more information, or because of external factors such as agents entering or leaving the market. Agents may also collude with others in hopes of securing better matches. Such considerations motivate the study of matchings that remain stable under preference changes.  

Suppose instance $A$ represents the current preferences and $B$ captures the changed instance. We are interested in matchings that are stable in both $A$ and $B$, these are called \emph{robust stable matchings}. A recent thread of work has considered the case where $A$ and $B$ are “nearby” instances—those differing only in the preference list of a {\em single agent}. Mai and Vazirani~\cite{MV.robust} and Gangam et al.~\cite{GMRV-nearby-instances} characterized the structure of stable matchings under such perturbations and gave efficient algorithms for identifying matchings that are stable under both instances. For this case, they showed that the  intersection of the sets of the stable matchings, $\mathcal{M}_A \cap \mathcal{M}_B$, is a sublattice of the lattices of the two instances, $\Lc_A$ and $\Lc_B$. They also provided efficient algorithms to find and enumerate these matchings.

Building on \cite{MV.robust}, Chen et al.~\cite{Chen-Matchings-Under-preferences}  introduced a model of robustness based on swap distance. A matching is defined to be $d$-robust if it remains stable under any $d$ swaps in the preference lists, even over multiple agents. They gave an efficient algorithm for deciding the existence of such matchings by exploiting the structure of the stable matching lattice. 

In this work, we allow multiple agents to make arbitrary changes and we study structural as well as algorithmic issues. Interestingly enough, the threshold at which structural properties break is very sharp, and we identify it.

\subsection{Our contributions}
\label{subsec.outlineresults}

From our perspective, the following is a central question. Its importance stems from the fact that the stable matching lattice is {\em universal} in the class of finite distributive lattices in the following sense: Given any lattice $\Lc$ from this class, there exists a stable matching instance, $A$, of an appropriate size, such that its lattice $\Lc_A$ is isomorphic to $\Lc$.  Consequently, an affirmative answer to the question below would imply a fundamental new property for the class of finite distributive lattices.

\begin{question}
\label{ques.one}
Is $\mathcal{M}_A \cap \mathcal{M}_B$ always a sublattice of $\Lc_A$ and $\Lc_B$ under arbitrary preference changes?
\end{question}

In Theorem~\ref{thm:both_sides_not_sublattice} we show that, in general, the answer is ``no''.  This leads to the natural question of identifying the cases in which this property does hold, and whether we can exploit it to obtain efficient algorithms. 

Consider a robust stable matching instance $(A,B)$ in which $p$ of the $n$ workers and $q$ of the $n$ firms change their preferences in going from $A$ to $B$. We show that the above statement fails when $(p, q) = (2, 2)$. 

In addition to structural questions, we study computational and geometric aspects of robust stable matchings:
\begin{itemize}
    \item Is the decision problem $\mathcal{M}_A \cap \mathcal{M}_B \neq \emptyset$ solvable in polynomial time?
    \item Can the worker-optimal and firm-optimal matchings in $\mathcal{M}_A \cap \mathcal{M}_B$ be efficiently computed?
    \item Does the set $\mathcal{M}_A \cap \mathcal{M}_B$ admit an efficient Birkhoff-partial order and support efficient enumeration?
    \item Is the intersection of stable matching polytopes integral, enabling linear programming(LP) based algorithms?
\end{itemize}

Our results, summarized in Figure~\ref{fig:summary_table}, characterize the structural and algorithmic behavior of $\mathcal{M}_A \cap \mathcal{M}_B$ as a function of $p$ and $q$, the number of workers and firms changing preferences. When $(p,q) = (n,1)$ or $(1,n)$—that is, any number of agents on one side and at most one on the other—the intersection remains a sublattice, the robust stable matching polytope remains integral, and both the worker-optimal and firm-optimal matchings in $\mathcal{M}_A \cap \mathcal{M}_B$ can be computed efficiently via combinatorial algorithms. In addition, LP-based methods can be used to find a robust stable matching in these cases. When changes are restricted to one side entirely, i.e., $(p,q) = (0,n)$ or $(n,0)$, we also obtain a succinct Birkhoff representation that supports enumeration with polynomial delay. While the partial order used for enumeration is efficiently computable in these one-sided cases, its computation \emph{remains open} for the two-sided $(1,n)$ or $(n,1)$ settings. Here, we show that it can be computed in polynomial time for $(1,n)$ if and only if it can be for $(1,1)$.

In contrast, we identify a sharp structural threshold at $(p,q) = (2,2)$: beyond this point, the intersection $\mathcal{M}_A \cap \mathcal{M}_B$ may no longer form a sublattice, and the integrality of the associated polytope can fail. For general $(p,q)$ instances, we present an $O(n^{p+q+2})$-time XP algorithm that decides whether a robust stable matching exists, constructs one if it does, and enumerates all such matchings with the same asymptotic delay. This yields a polynomial-time algorithm and enumeration when the number of agents with differing preferences is bounded by a constant $k$ (i.e., $p + q = k$). Determining the precise threshold for the computational complexity in the general case remains an intriguing open question.
\begin{figure}[ht]
\centering
\renewcommand{\arraystretch}{1.3}
\begin{tabular}{|c|c|c|c|}
\hline
\textbf{$(p, q)$} & \textbf{Computation (P?)} & \textbf{Structure (Lattice?)} & \textbf{Geometry   (Integral?)} \\ \midrule
(0,1) & P~\cite{GMRV-nearby-instances} & Yes~\cite{GMRV-nearby-instances} & Yes~[Thm.~\ref{thm:n_1_integral_polytope}] \\
(0,n) & P~[Thm.~\ref{thm:n_1_algo_works}] & Yes~[Thm.~\ref{thm:birkhoff}] & Yes~[Thm.~\ref{thm:n_1_integral_polytope}] \\
(1,1) & P~[Thm.~\ref{thm:n_1_algo_works}] & \cellcolor{yellow!25}Yes~[Thm.~\ref{thm:sublattice_two_side}] & Yes~[Thm.~\ref{thm:n_1_integral_polytope}] \\
(1,n) & P~[Thm.~\ref{thm:n_1_algo_works}] & \cellcolor{yellow!25}Yes~[Thm.~\ref{thm:sublattice_two_side}] & Yes~[Thm.~\ref{thm:n_1_integral_polytope}] \\
(2,2) & P~[Thm.~\ref{thm:robust_xp}] & No~[Thm.~\ref{thm:both_sides_not_sublattice}] & No~[Thm.~\ref{thm:2_2_not_integral_polytope}] \\
($p$,$q$) & $O(n^{p+q+2})$~[Thm.~\ref{thm:robust_xp}] & - & - \\
($n$,$n$) & NP-Complete~\cite{NP-Two-stable} & No~[Thm.~\ref{thm:both_sides_not_sublattice}] & No~[Thm.~\ref{thm:2_2_not_integral_polytope}] \\
\hline
\end{tabular}
\vspace{0.5em}
\caption{
Summary of our results on robust stable matchings. 
$p$ workers and $q$ firms permute their preference lists from instance $A$ to $B$.\\
The ``Computation'' column deals with the decision problem $\mathcal{M}_A \cap \mathcal{M}_B \neq \varnothing $. \\
The ``Structure'' column answers if $\mathcal{M}_A \cap \mathcal{M}_B$ is a sublattice of $\Lc_A$ and $\Lc_B$. \\
The ``Geometry'' column answers if the robust fractional stable matching polytope is integral. 
}
\label{fig:summary_table}
\end{figure}





\section{Related Work}
\label{sec:related_work}

The stable matching problem was introduced by Gale and Shapley~\cite{GaleS}, who also proposed the celebrated Deferred Acceptance (DA) algorithm. The DA algorithm computes a stable matching that is optimal for one side of the market and pessimal for the other. The set of stable matchings forms a distributive lattice, as shown in~\cite{knuth1976marriages}, and this structure has played a central role in many algorithmic developments. Key game-theoretic properties such as the Rural Hospitals Theorem~\cite{Roth-IC-1982economics} and incentive compatibility of the DA algorithm~\cite{DubinsF} have also been established. The integrality of the stable matching polytope was shown by Teo and Sethuraman~\cite{Teo-Sethuraman-Lp}, with alternative formulations given in~\cite{VANDEVATE, Roth85}. For comprehensive treatments of these topics, we refer the reader to~\cite{Knuth, GusfieldI, Manlove-book}.

Mai and Vazirani~\cite{MV.robust} introduced the notion of robustness used in this paper, defining a matching to be robust if it remains stable under both the original and a perturbed instance. They provided polynomial-time algorithms for the case where a single agent modifies preferences via a simple downshift. This was generalized by Gangam et al.~\cite{GMRV-nearby-instances} to arbitrary changes by one agent. Chen et al.~\cite{Chen-Matchings-Under-preferences} introduced the notion of $d$-robust stable matchings, where up to $d$ swaps are allowed in the preference profiles, and also studied the trade-offs between social welfare and stability. Miyazaki and Okamoto~\cite{NP-Two-stable} have shown that when all agents are allowed to change preferences from one instance to another, finding a robust stable matching---which they term a jointly stable matching---is NP-hard.

Building on this work, we consider the case where multiple agents on both sides simultaneously change their preferences via arbitrary permutations. While earlier approaches were largely based on lattice structure, rotation posets, or variants of the DA algorithm, our approach also incorporates linear programming (LP) techniques into the analysis.

When only one side of the market changes preferences, the setting is closely related to stable matching with strict preferences on one side and weak or partial preferences on the other~(see Appendix~\ref{app:one_side}). In such settings, several refined notions of stability have been studied, including weak, strong, and super-strong stable matchings. Irving~\cite{IRVING-Indifferences} showed that weakly stable matchings always exist and can be computed in polynomial time. Spieker~\cite{SPIEKER-indifference-lattice} showed that the set of super-stable matchings forms a distributive lattice. Manlove~\cite{MANLOVE-indifference-lattice} extended this to strong stability, showing that the corresponding set of matchings also retains a lattice structure. Kunysz et al.~\cite{Kunysz-ssm} gave an efficient characterization of strongly stable matchings via succinct partial orders.

Aziz et al.~\cite{aziz1, aziz2} studied uncertain preferences and proposed robust solutions that perform well across all completions of partial orders. Genc et al.~\cite{genc2, genc1} proposed $(x, y)$-supermatches, which allow re-stabilizing the matching after any $x$ agents break up by rematching them while affecting at most $y$ other pairs. While not robust in our sense, they provide a notion of ease of repair.

Chen et al.~\cite{NP-Two-stable-Chen} introduced stable matching under multi-modal preferences, where each agent evaluates potential matches using multiple criteria. Menon and Larson~\cite{Menon_Larson} addressed the problem of minimizing the maximum number of blocking pairs across all completions of weak orders. Incremental changes to preferences over time have also been studied. Bredereck et al.~\cite{BredereckCKLN20} and Boehmer et al.~\cite{boehmer_incremental_mfcs, boehmer_incremental_aaai} explored the complexity of adapting stable matchings as inputs evolve, providing structural results and efficient algorithms in restricted settings.

Popularity is a voting-based relaxation of stability first introduced by Gärdenfors~\cite{Gard75a}. A matching is popular if it does not lose a head-to-head election against any other matching. Popular matchings have been extensively studied. Recently, Bullinger et al.~\cite{robust_pop_matchings_BGS} introduced the notion of robust popular matchings and gave polynomial-time and hardness results based on the extent of changes in preferences. Csáji~\cite{Csaj24a} studied robust popular matchings in the presence of multi-modal and uncertain preferences.

Robustness to input perturbations has also been studied in voting theory~\cite{FaRo15a, SYE13a, BFK+21a}, where it is often interpreted through the lens of swap-bribery. The cost of altering votes is typically measured via the swap distance~\cite{EFS09a}. Similar ideas have been applied to stable matchings~\cite{BBHN21a}, where preference manipulations are used to enforce or prevent certain matches. Bérczi et al.~\cite{BERCZI} study how one can deliberately change preferences to enforce the existence of stable matchings with desired properties.


\section{Preliminaries}
\label{sec.prelim}

\subsection{The stable matching problem and the lattice of stable matchings}
\label{subsection.latticeOfSM}

A stable matching problem instance consists of a set of $n$ workers, $\calW = \{w_1, w_2, \ldots, w_n\}$, and a set of $n$ firms, $\calF = \{f_1, f_2, \ldots, f_n\}$, collectively referred to as agents. Each agent $a \in \calW \cup \calF$ has a preference profile $>_a$, which is a strict total order over the agents of the opposite type. For example, $w_i >_f w_j$ indicates that firm $f$ strictly prefers worker $w_i$ to $w_j$. Worker preferences are expressed analogously.

A matching $M$ is a one-to-one correspondence between $\calW$ and $\calF$. For each pair $(w, f) \in M$, $w$ is called the partner of $f$ in $M$ (or $M$-partner), and vice versa. A pair $(w, f) \notin M$ is a \emph{blocking pair} if $f >_w M(w)$ and $w >_f M(f)$; that is, if both strictly prefer each other to their respective partners in $M$. A matching is \emph{stable} under the instance if it admits no blocking pairs.

Let $M$ and $M'$ be two stable matchings. We say that $M$ \emph{dominates} $M'$, denoted $M \preceq M'$, if every worker weakly prefers their partner in $M$ to their partner in $M'$. In this case, $M$ is also called a \emph{predecessor} of $M'$. A stable matching $M$ is a \emph{common predecessor} of two stable matchings $M_1$ and $M_2$ if it is a predecessor of both. It is a \emph{lowest common predecessor} of $M_1$ and $M_2$ if no other common predecessor $M'$ satisfies $M \preceq M'$. Analogously, one can define the notions of \emph{successor} and \emph{highest common successor}.  

This dominance partial order has the following key property: for any two stable matchings $M_1$ and $M_2$, their lowest common predecessor and highest common successor are unique. That is, the set of stable matchings forms a \emph{lattice} under this partial order. The lowest common predecessor is called the \emph{meet}, denoted $M_1 \wedge M_2$, and the highest common successor is called the \emph{join}, denoted $M_1 \vee M_2$.

One can show that $M_1 \wedge M_2$ is the matching obtained when each worker chooses their more preferred partner from $M_1$ and $M_2$; it is easy to verify that this matching is also stable. Interestingly, $M_1 \wedge M_2$ also results when each firm chooses their less preferred partner from $M_1$ and $M_2$.  

Similarly, $M_1 \vee M_2$ is the matching in which each worker (respectively, firm) chooses their less (respectively, more) preferred partner from $M_1$ and $M_2$, and this matching is also stable.  

This duality implies that all the above definitions and lattice properties hold symmetrically when phrased in terms of firms instead of workers. Moreover, the lattice operations \emph{join} and \emph{meet} distribute: given three stable matchings $M, M', M''$,
\begin{align*}
    M \vee (M' \wedge M'') &=  (M \vee M') \wedge (M \vee M'') \\
    M \wedge (M' \vee M'') &= (M \wedge M') \vee (M \wedge M'').
\end{align*}

The lattice of stable matchings contains a unique matching $M_{\top}$ that dominates all others, and a unique matching $M_{\bot}$ that is dominated by all others. The matching $M_{\top}$ is called the \emph{worker-optimal stable matching}, as each worker is matched to their most preferred firm among all stable matchings. This is also the \emph{firm-pessimal stable matching}. Similarly, $M_{\bot}$ is the \emph{worker-pessimal} and \emph{firm-optimal stable matching}. Since the number of stable matchings under any instance is finite, the stable matching lattice is a \emph{finite distributive lattice}.

\subsection{Birkhoff’s Theorem, sublattices and compressions}
\label{sec.birkhoff}

\begin{definition}
\label{def:closedset}
A closed set of a partially ordered set (poset) is a subset $S$ such that, if an element is in $S$, then all of its predecessors are also in $S$.
\end{definition}

The collection of closed sets (also called lower sets) of a partial order $\Pi$ is closed under union and intersection, and forms a distributive lattice, with join and meet corresponding to these two operations, respectively. We denote this lattice by $L(\Pi)$.  
Birkhoff's theorem~\cite{Birkhoff}, also known as the \emph{fundamental theorem for finite distributive lattices} (e.g., see~\cite{Stanley}), states that for any finite distributive lattice $\Lc$, there exists a partial order $\Pi$ such that $L(\Pi) \cong \Lc$, i.e., the lattice of closed sets of $\Pi$ is isomorphic to $\Lc$. We say that $\Pi$ \emph{generates} $\Lc$.

\begin{theorem}[Birkhoff~\cite{Birkhoff}]
\label{thm:Birkhoff}
Every finite distributive lattice $\Lc$ is isomorphic to $L(\Pi)$, for some finite poset $\Pi$.
\end{theorem}

A \emph{join semi-sublattice} $\Lc_j$ of a distributive lattice $\Lc$ is a subset such that, for any $x, y \in \Lc_j$, the join $x \vee y$ also belongs to $\Lc_j$.  
Similarly, a \emph{meet semi-sublattice} $\Lc_m$ is a subset such that, for any $x, y \in \Lc_m$, the meet $x \wedge y$ also belongs to $\Lc_m$.  
A \emph{sublattice} $\Lc'$ of $\Lc$ is a subset that is both a join and meet semi-sublattice.

The \emph{Hasse diagram} of a poset is a directed graph with a vertex for each element in the poset and an edge from $x$ to $y$ if and only if $x \prec y$ and there is no $z$ such that $x \prec z \prec y$, i.e., all precedence relations implied by transitivity are suppressed.

Let $\Lc$ be the lattice generated by a poset $\Pi$, and let $H(\Pi)$ denote the Hasse diagram of $\Pi$. To derive a new poset $\Pi'$, consider the following operations: choose a set $E$ of directed edges and add them to $H(\Pi)$. Let $H_E$ be the resulting graph. Let $H'$ be the graph obtained by shrinking the strongly connected components of $H_E$. Define $\Pi'$ to be the partial order induced by the non-shrunk edges of $H'$ on its nodes. The poset $\Pi'$ is called a \emph{compression} of $\Pi$.

\begin{theorem}[\cite{GMRV-nearby-instances}, Theorem 1]
\label{thm:generalization} 
There is a one-to-one correspondence between the compressions of $\Pi$ and the sublattices of $L(\Pi)$. Moreover, if a sublattice $\Lc' \subseteq L(\Pi)$ corresponds to a compression $\Pi'$, then $\Lc'$ is generated by $\Pi'$.
\end{theorem}

Let $\Lc'$ be the sublattice of $\Lc$ generated by $\Pi'$. We say that a set of edges $E$ \emph{defines} $\Lc'$.\label{def:edge} Note that while a given edge set $E$ uniquely determines the sublattice $L(\Pi')$, multiple such edge sets may define the same sublattice.

One way to define a partial order for the stable matching lattice is through the concept of \emph{rotations}. The partial order on rotations can be computed efficiently and admits a succinct representation, even when the lattice itself is exponentially large~\cite{irving2}. As described in~\cite{GusfieldI}, the rotation poset supports efficient enumeration of all stable matchings.

Compressions defined on the rotation poset characterize sublattices of the original lattice by grouping rotations into \emph{meta-rotations} and imposing a partial order on them. By using the poset of meta-rotations instead of the full rotation poset, one can efficiently enumerate matchings in the corresponding \emph{sublattice}. (Additional details about the rotation poset appear in \Cref{app:rotations}, though they are not needed for the main results of this paper.)

\subsection{Robust Stable Matchings}

A robust stable matching instance consists of $k$ stable matching instances $A_1, A_2, \ldots, A_k$ with $k \geq 2$, each defined on the same set of $2n$ agents: $n$ workers $\calW = \{w_1, w_2, \ldots, w_n\}$ and $n$ firms $\calF = \{f_1, f_2, \ldots, f_n\}$, where all preference lists are strict and complete.

We say that the instances are of type $(p, q)$ (or just are $(p,q)$) if there exists an instance $A_i$ such that for every $j \in \{1, 2, \ldots, k\}$, the preferences in $A_j$ differ from those in $A_i$ for at most $p$ workers and $q$ firms, and remain the same for the other $2n - p - q$ agents. The instance $A_i$ is referred to as the \emph{original} instance, and the others are called the \emph{changed} instances.

\begin{definition}
    Given a robust stable matching instance $(A_1, A_2, \ldots, A_k)$, a matching is said to be \emph{robust stable} under these instances if it is stable under each instance $A_i$ for all $i \in \{1, 2, \ldots, k\}$.
\end{definition}

We focus primarily on the case where $k = 2$, and let $(A, B)$ be a robust stable matching instance of type $(p, q)$, defined on $\calW = \{w_1, w_2, \ldots, w_n\}$ and $\calF = \{f_1, f_2, \ldots, f_n\}$. Here, $A$ is the original instance and $B$ is the changed instance. The results are symmetric in $p$ and $q$.

For any stable matching instance $I$, let $\calM_I$ denote the set of all stable matchings and $\calL_I$ the corresponding lattice of stable matchings. Given a pair $(A, B)$ of type $(p, q)$, our goal is to understand the structure and computational tractability of the robust stable matching set $\calM_A \cap \calM_B$, i.e., the matchings that are stable in both $A$ and $B$. Specifically, we investigate whether $\calM_A \cap \calM_B$ forms a sublattice of $\calL_A$ (and of $\calL_B$), whether such a matching can be found efficiently, and whether linear programming techniques can be leveraged to compute one.

\subsection{Lattices of nearby instances}
\label{sec.lattice}

Consider two instances $A$ and $B$ that are of type $(0,1)$ or $(1,0)$. As the number of matchings stable under an instance can be exponential in $n$, checking the stability of every matching in $\mathcal{M}_A$ with respect to $B$ may require exponential time. The results of~\cite{GMRV-nearby-instances}, which reveal the structure of these robust stable matchings, are therefore essential and form the foundation for our work. We recall them below (statements are modified to match our notation).

\begin{theorem}[\cite{GMRV-nearby-instances}, Theorem 29]
\label{thm:bouquet}
If a lattice $\calL$ can be partitioned into a sublattice $\calL_1$ and a semi-sublattice $\calL_2$, and there is a polynomial-time oracle that determines whether any $x \in \calL$ belongs to $\calL_1$ or $\calL_2$, then an edge set defining $\calL_1$ can be found in polynomial time.
\end{theorem}

\begin{lemma}[\cite{GMRV-nearby-instances}]
\label{lem:1_0_sublattice}
If instances $A$ and $B$ are of type $(0,1)$ ( or $(1,0)$), then:
\begin{enumerate}
    \item $\calL' = \calM_A \cap \calM_B$ is a sublattice of both $\calL_A$ and $\calL_B$.
    \item $\calM_A \setminus \calM_B$ is a join (or meet) semi-sublattice of $\calL_A$. 
    \item A set of edges defining the sublattice $\calL' = \calM_A \cap \calM_B$ can be computed in polynomial time.
\end{enumerate}
\end{lemma}


These results can be extended to multiple instances in the following way. Let $A, B_1, B_2, \ldots, B_k$ be such that each instance $B_i$ is formed from $A$ by arbitrarily permuting the preferences of one agent (possibly a different agent for each $B_i$), i.e., for all $i \in \{1, 2, \ldots, k\}$, the pair $(A, B_i)$ is of type $(1,0)$ or $(0,1)$. \Cref{lem:1_0_sublattice} tells us that each $\calL_i = \calM_A \cap \calM_{B_i}$ is a sublattice of $\calL_A$. Let the edge set defining $\calL_i$ be $E_i$, and define $\calL' = \calM_A \cap \calM_{B_1} \cap \calM_{B_2} \cap \ldots \cap \calM_{B_k}$. Then:
\begin{lemma}
\label{lem:1_0_multiple}
\begin{enumerate}
    \item (\cite{GMRV-nearby-instances}) $\calL'$ is a sublattice of $\calL_A$.
    \item (\cite{GMRV-nearby-instances}) $E = \bigcup_i E_i$ defines $\calL'$.
    \item (\cite{GMRV-nearby-instances}) $E$, and hence the partial order $\Pi'$ generating $L'$, can be computed in polynomial time.
    \item (\cite{GusfieldI}) Matchings in $\calL'$ can be enumerated efficiently.
\end{enumerate}
\end{lemma}




\subsection{Linear programming formulation}
\label{sec:lp_intro}

\begin{figure}[htp]
    \centering
    \includegraphics[width=8cm]{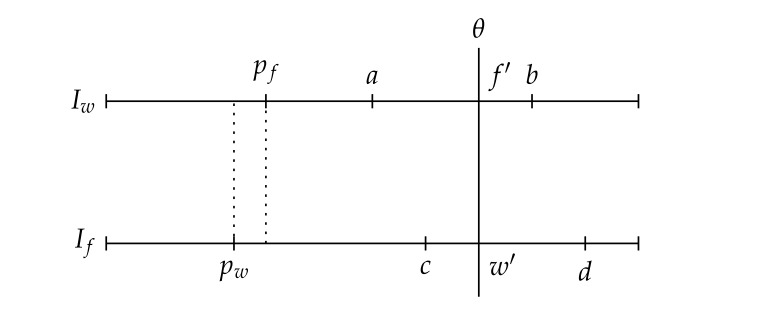}
    \caption{Finding integral stable matching from a fractional solution (Image from~\cite{Book-Online})}
    \label{fig:1}
\end{figure} 

A well-known result about the stable matching problem is that it admits a linear programming formulation in which all optimal vertices are integral(\cite{Teo-Sethuraman-Lp,VANDEVATE}). This yields an efficient method to compute stable matchings. The LP formulation is given below as~\eqref{eq:lp_single}. 

The third constraint disallows blocking pairs: for any worker-firm pair $(w, f)$, it ensures that if $w$ is matched to someone less preferred than $f$, then $f$ must be matched to someone more preferred than $w$. The remaining constraints ensure that $x$ defines a fractional perfect matching.

\begin{equation}
\label{eq:lp_single}
    \begin{aligned}[b]
        \text{maximize} \quad & 0 \\
        \text{subject to} \quad 
        & \sum_w x_{wf} = 1 \quad && \forall f \in \calF, \\
        & \sum_f x_{wf} = 1 \quad && \forall w \in \calW, \\
        & \sum_{f >^A_w f'} x_{wf'} - \sum_{w' >^A_f w} x_{w'f} \leq 0 \quad && \forall w \in \calW,\ \forall f \in \calF, \\
        & x_{wf} \geq 0 \quad && \forall w \in \calW,\ \forall f \in \calF.    
    \end{aligned}
\end{equation}


A solution $x$ to LP~\ref{eq:lp_single} can be converted into an integral perfect matching as follows. Construct $2n$ unit intervals, each corresponding to an agent. For each worker $w$, divide its interval $I_w$ into $n$ subintervals of lengths $x_{wf}$ (note that $\sum_f x_{wf} = 1$), ordered from $w$'s most preferred firm to their least. Apply the same process to each firm $f$'s interval $I_f$, but order the subintervals from $f$'s least to most preferred worker.

Choose a value $\theta$ uniformly at random from $[0,1]$, and identify the subinterval in each agent’s interval that contains $\theta$. The probability that $\theta$ lies exactly on a boundary is zero, so we may ignore this case.

Let $\mu_{\theta} : \calW \rightarrow \calF$ denote the firm corresponding to the subinterval containing $\theta$ in each worker's interval, and let $\mu'_{\theta} : \calF \rightarrow \calW$ denote the analogous mapping for firms. These functions define perfect matchings and are inverses of each other. Moreover, they are stable for any $\theta \in [0,1]$.

\begin{lemma}[\cite{Teo-Sethuraman-Lp}]
\label{lem:lp_perfect} 
If $\mu_{\theta}(w) = f$ then $\mu'_{\theta}(f) = w$.
\end{lemma}

\begin{lemma}[\cite{Teo-Sethuraman-Lp}]
\label{lem:lp_stable} 
For each $\theta \in [0,1]$, the matching $\mu_{\theta}$ is stable.
\end{lemma}

\begin{theorem}[\cite{VANDEVATE}]
\label{thm:lp_polytope}
The stable matching polytope defined by LP~\ref{eq:lp_single} has integral optimal vertices; that is, it is the convex hull of stable matchings.
\end{theorem}

\section{Results}
\label{sec:results}

In this section, we present our results. We primarily consider the case of $(p, q)$ robust stable matchings between two instances $(A, B)$ defined on $\mathcal{W} = \{w_1, w_2, \ldots, w_n\}$ and $\mathcal{F} = \{f_1, f_2, \ldots, f_n\}$. We assume $0 \leq p \leq q \leq n$, since the results are symmetric in $p$ and $q$.

In \Cref{sec:lattice}, we investigate the lattice structure of robust stable matchings for various values of $p$ and $q$. For cases where the lattice structure is preserved, \Cref{sec:optimal_matchings} provides algorithms to compute worker-optimal and firm-optimal robust stable matchings. In \Cref{sec:polytope}, we use linear programming techniques to obtain robust stable matchings by characterizing when the corresponding polytope admits integral optimal vertices. Finally, in \Cref{sec:robust_general_algo}, we address the general setting and present an XP-time algorithm that decides the existence of a robust stable matching. We would like to note that, while the results focus primarily on two stable matching instances, they can all be extended to many instances(see \Cref{thm:n_1_multiple}). (Proofs of all theorems and lemmas marked ${}^\dagger$ appear in \Cref{app:theorem_proofs} and \Cref{app:algo_proofs}.)


\subsection{Lattice structure}
\label{sec:lattice}

The primary question we address is whether the set of robust stable matchings for $(A, B)$, namely $\mathcal{M}_A \cap \mathcal{M}_B$, always forms a sublattice of both $\Lc_A$ and $\Lc_B$. While \cite{GMRV-nearby-instances} showed that this set is a sublattice in the $(0,1)$ case, their proof relies on the fact that $(\mathcal{M}_A \setminus \mathcal{M}_B)$ forms a semi-sublattice of $\Lc_A$. They also prove that as long as changes are restricted to one side—i.e., when $(A, B)$ is of type $(0,n)$ or $(n,0)$—the set $\mathcal{M}_A \cap \mathcal{M}_B$ remains a sublattice of both $\Lc_A$ and $\Lc_B$~\cite[Proposition 6]{GMRV-nearby-instances}.

A trial-and-error search for counterexamples when agents on both sides change preferences yielded surprising results. In most instances examined, the sublattice property still held. As expected, when only one side changes preferences, the join and meet operations in $\Lc_A$ and $\Lc_B$ coincide. However, when both sides change preferences, this alignment no longer holds. Surprisingly, even in such cases, the set $\mathcal{M}_A \cap \mathcal{M}_B$ often remained a sublattice under both $\Lc_A$ and $\Lc_B$. Examples illustrating this behavior are shown in \Cref{app:examples}.

However, as \Cref{thm:both_sides_not_sublattice} shows, the intersection $\mathcal{M}_A \cap \mathcal{M}_B$ is not always a sublattice. An explicit example is provided in \Cref{ex:twoSidesNotSublattice}. Thus, while the sublattice property holds for $(0,n)$, it no longer holds for $(2,2)$.

\begin{restatable}{theorem}{bothsidesnotsublattice}
${}^\dagger$
\label{thm:both_sides_not_sublattice}
When $(A, B)$ are of type $(p, q)$ with $p, q \geq 2$, the set $\mathcal{M}_A \cap \mathcal{M}_B$ is not always a sublattice of $\Lc_A$ and $\Lc_B$.
\end{restatable}

This raises the question: where does the sublattice property break down? We show that as long as no more than one agent changes preferences on one side, the sublattice property still holds. In particular, we establish that for $(1,n)$, the intersection $\mathcal{M}_A \cap \mathcal{M}_B$ remains a sublattice.

\begin{figure}[ht]
	\begin{wbox}
	\centering
	\begin{minipage}{.45\linewidth}
        \centering
		\begin{tabular}{p{0.5cm}|p{0.5cm}p{0.5cm}p{0.5cm}p{0.5cm}}
			1  & a & b & c & d \\
			2  & b & c & a & d \\
			3  & d & c & a & b \\
			4  & c & d & a & b
		\end{tabular}
		
		\hspace{1cm}
		
		Worker preferences in $A$ 
		
		\hspace{1cm}
		
	\end{minipage}%
	\begin{minipage}{.45\linewidth}
    \centering
	\begin{tabular}{p{0.5cm}|p{0.5cm}p{0.5cm}p{0.5cm}p{0.5cm}}
		a  & 2 & 4 & 1 & 3 \\
		b  & 1 & 2 & 3 & 4 \\
		c  & 3 & 4 & 2 & 1 \\
		d  & 4 & 3 & 2 & 1
	\end{tabular}
	
	\hspace{1cm}
	
	Firm preferences in $A$
	 
	\hspace{1cm}
\end{minipage}%

	\begin{minipage}{.45\linewidth}
        \centering
		\begin{tabular}{p{0.5cm}|p{0.5cm}p{0.5cm}p{0.5cm}p{0.5cm}}
		1  & a & b & c & d \\
		2  & b & c & a & d \\
		\cellcolor{red!25}3  & \cellcolor{red!25}c & \cellcolor{red!25}d & \cellcolor{red!25}a & \cellcolor{red!25}b \\
		\cellcolor{red!25}4  & \cellcolor{red!25}d & \cellcolor{red!25}a & \cellcolor{red!25}c & \cellcolor{red!25}b 
		\end{tabular}

		\hspace{1cm}

		Worker preferences in $B$
	\end{minipage} 
	\begin{minipage}{.45\linewidth}
        \centering
		\begin{tabular}{p{0.5cm}|p{0.5cm}p{0.5cm}p{0.5cm}p{0.5cm}}
		a  & 2 & 4 & 1 & 3 \\
		b  & 1 & 2 & 3 & 4 \\
		\cellcolor{red!25}c  & \cellcolor{red!25}4 & \cellcolor{red!25}2 & \cellcolor{red!25}3 & \cellcolor{red!25}1 \\
		\cellcolor{red!25}d  & \cellcolor{red!25}3 & \cellcolor{red!25}4 & \cellcolor{red!25}1 & \cellcolor{red!25}2 
		\end{tabular}

		\hspace{1cm}

		Firm preferences in $B$
	\end{minipage} 
	\end{wbox}
        \caption{$\mathcal{M}_A \cap \mathcal{M}_B$ is not a sublattice of $\Lc_A$. Preference lists are in decreasing order. The agents whose preferences change from $A$ to $B$ are marked in red. This notation is used for all examples.}
	\label{ex:twoSidesNotSublattice}
\end{figure}

Let $A$ and $B$ be a $(1,n)$ robust stable matching instance where only $w_1$ and all firms change their preferences from $A$ to $B$.

\begin{restatable}{theorem}{sublatticetwoside}
\label{thm:sublattice_two_side}
If $A$ and $B$ are of type $(1,n)$, then $\mathcal{M}_A \cap \mathcal{M}_B$ is a sublattice of both $\Lc_A$ and $\Lc_B$. 
\end{restatable}

\begin{proof}
Without loss of generality, assume $|\mathcal{M}_A \cap \mathcal{M}_B| > 1$, and let $M_1$ and $M_2$ be two distinct matchings in $\mathcal{M}_A \cap \mathcal{M}_B$. Let $\vee_A$ and $\vee_B$ denote the join operations under $A$ and $B$, respectively. Likewise, let $\wedge_A$ and $\wedge_B$ denote the meet operations under $A$ and $B$.

The join $M_1 \vee_A M_2$ is obtained by assigning each worker to their less preferred partner from $M_1$ and $M_2$, according to instance $A$. Since the preferences of workers $w_2, w_3, \ldots, w_n$ are identical in $A$ and $B$, their less preferred partners under $A$ and $B$ coincide. Thus, for all $w \neq w_1$, we have:
$$(M_1 \vee_A M_2)(w) = (M_1 \vee_B M_2)(w).$$

Since $M_1 \vee_A M_2$ and $M_1 \vee_B M_2$ are both stable matchings under their respective instances, they are also perfect matchings. Hence, the remaining worker $w_1$ must be matched to the same firm in both matchings:
$$(M_1 \vee_A M_2)(w_1) = (M_1 \vee_B M_2)(w_1).$$

Therefore, $M_1 \vee_A M_2 = M_1 \vee_B M_2$. A similar argument shows that $M_1 \wedge_A M_2 = M_1 \wedge_B M_2$. Hence, the join and meet operations under instances $A$ and $B$ are equivalent, and both $M_1 \vee_A M_2$ and $M_1 \wedge_A M_2$ belong to $\mathcal{M}_A \cap \mathcal{M}_B$. The theorem follows.
\end{proof}

Thus, when the instances are of type $(1,n)$, \Cref{thm:generalization} guarantees the existence of a poset that generates the corresponding sublattice. However, unlike the $(0,1)$ case, \Cref{thm:bouquet} does not apply to instances of type $(0,n)$.

\begin{restatable}{lemma}{lemnnotsemisublattice}
${}^\dagger$
\label{lem:0_n_not_semi_sublattice}
If $A$ and $B$ are of type $(0,n)$, then $\mathcal{M}_A \setminus \mathcal{M}_B$ is not always a semi-sublattice of $\Lc_A$.
\end{restatable}

\Cref{ex:oneSideNotSemiSubLattice} provides an example where $A$ and $B$ are $(0,n)$, but $\mathcal{M}_A \setminus \mathcal{M}_B$ is not a semi-sublattice of $\Lc_A$, leading to the above lemma.

\begin{figure}[ht]
	\begin{wbox}
	\begin{minipage}{.37\linewidth}
        \centering
		\begin{tabular}{l|lllll}
			1  & a & b & d & c & e \\
			2  & b & a & c & d & e\\
			3  & c & d & b & a & e\\
			4  & b & d & c & a & e\\
                5  & c & e & a & b & d
		\end{tabular}
		
		\hspace{1cm}
		
		Worker preferences in $A$ and $B$
		
		\hspace{1cm}
		
	\end{minipage}%
	\begin{minipage}{.33\linewidth}
        \centering
	\begin{tabular}{l|lllll}
		a  & 4 & 2 & 1 & 3 & 5 \\
		b  & 3 & 1 & 2 & 4 & 5\\
		c  & 2 & 4 & 3 & 1 & 5\\
		d  & 1 & 3 & 4 & 2 & 5 \\
            e  & 5 & 1 & 2 & 3 & 4
	\end{tabular}
	
	\hspace{1cm}
	
	Firm preferences in $A$
	 
	\hspace{1cm}
 
        \end{minipage}%
	\begin{minipage}{.33\linewidth}
        \centering
		\begin{tabular}{l|lllll}
		a  & 4 & 2 & 1 & 3 & 5 \\
		\cellcolor{red!25}b  & \cellcolor{red!25}1 & \cellcolor{red!25}4 & \cellcolor{red!25}2 & \cellcolor{red!25}3 & \cellcolor{red!25}5 \\
            \cellcolor{red!25}c  & \cellcolor{red!25}4 & \cellcolor{red!25}5 & \cellcolor{red!25}3 & \cellcolor{red!25}1 & \cellcolor{red!25}2 \\
            d  & 1 & 3 & 4 & 2 & 5 \\
            e  & 5 & 1 & 2 & 3 & 4 
		\end{tabular}

		\hspace{1cm}

		Firm preferences in $B$

            \hspace{1cm}
	\end{minipage} 
	\end{wbox}
	\caption{$\mathcal{M}_{A} \setminus \mathcal{M}_B$ is not a semi-sublattice of $\Lc_A$.}
	\label{ex:oneSideNotSemiSubLattice} 
\end{figure} 

This shows that \Cref{thm:bouquet} cannot be used directly to find Birkhoff’s partial order in the $(0,n)$ setting. To circumvent this, we define $n$ hybrid instances $B_1, \dots, B_n$ such that for each $i \in \{1, 2, \ldots, n\}$, the pair $(A, B_i)$ is of type $(0,1)$. Specifically, the preference profile of agent $x$ in instance $B_i$ is defined as:
\begin{equation*}
>_x^{B_i} = 
\left \{
    \begin{array}{ll}
        >_x^B, & \text{if } x = f_i \\
        >_x^A, & \text{if } x \neq f_i
    \end{array}
\right\}
\end{equation*}
That is, instance $B_i$ is identical to instance $A$, except that the preferences of firm $f_i$ are replaced by those in instance $B$.

Defined this way, \Cref{thm:breaking_instances_one_side} shows that the matchings stable under both $A$ and $B$ are precisely those stable under $A$ and all the hybrid instances $B_1, B_2, \ldots, B_n$.

\begin{restatable}{theorem}{breakinginstancesoneside}
${}^\dagger$
\label{thm:breaking_instances_one_side}
$\mathcal{M}_A \cap \mathcal{M}_B = \mathcal{M}_A \cap \mathcal{M}_{B_1} \cap \mathcal{M}_{B_2} \cap \ldots \cap \mathcal{M}_{B_n}$.
\end{restatable}

\begin{restatable}{theorem}{birkhoff}
\label{thm:birkhoff}
If $A$ and $B$ are of type $(0,n)$, then Birkhoff’s partial order generating $\mathcal{M}_A \cap \mathcal{M}_B$ can be computed efficiently, and the matchings in this set can also be enumerated with polynomial delay.
\end{restatable}

\begin{proof}
Since each $B_i$ differs from instance $A$ at exactly one agent, we apply \Cref{lem:1_0_sublattice} to compute the edge sets $E_i$ defining $\mathcal{M}_A \cap \mathcal{M}_{B_i}$. By \Cref{lem:1_0_multiple}, the union $E = \bigcup_i E_i$ defines the sublattice $\mathcal{M}_A \cap \mathcal{M}_B$, and can be used to enumerate all matchings in this intersection efficiently.
\end{proof}

A similar technique can be attempted to the case where $(A, B)$ is of type $(1,n)$. Let $w_1 \in \calW$ and $f_1, \dots, f_n \in \calF$ be the worker and firms whose preferences differ between $A$ and $B$. For any subset $X \subseteq \calW \cup \calF$, let $B_X$ denote the instance in which each agent in $X$ has the same preferences as in $B$, and all other agents have the same preferences as in $A$. Then:

\begin{restatable}{theorem}{breakinginstancestwoside}
${}^\dagger$
\label{thm:breaking_instances_two_side}
$\mathcal{M}_A \cap \mathcal{M}_B = \mathcal{M}_A \cap \mathcal{M}_{B_{\{w_1, f_1\}}} \cap \mathcal{M}_{B_{\{w_1, f_2\}}} \cap \dots \cap \mathcal{M}_{B_{\{w_1, f_n\}}}$.
\end{restatable}

\begin{restatable}{corollary}{birkhoffonen}
${}^\dagger$
\label{cor:birkhoff_one_n}
Computing Birkhoff’s partial order for the $(1,n)$ case is in $\mathcal{P}$ if and only if computing it for the $(1,1)$ case is in $\mathcal{P}$.
\end{restatable}

Unlike the $(0,1)$ case, the set $\mathcal{M}_A \setminus \mathcal{M}_B$ is not always a semi-sublattice when $(A, B)$ is of type $(1,1)$; see the example in \Cref{ex:1-1notSemisublattice}. As a result, \Cref{thm:bouquet} cannot be applied, and characterizing the sublattice structure remains an open problem.

\begin{restatable}{lemma}{lnotsemisublattice}
${}^\dagger$
\label{lem:1_1_not_semi_sublattice}
If $A$ and $B$ are of type $(1,1)$, then $\mathcal{M}_A \setminus \mathcal{M}_B$ is not always a semi-sublattice of $\Lc_A$.
\end{restatable}

\begin{figure}[ht]
	\begin{wbox}
	\centering
	\begin{minipage}{.45\linewidth}
        \centering
		\begin{tabular}{p{0.5cm}|p{0.5cm}p{0.5cm}p{0.5cm}p{0.5cm}p{0.5cm}}
			1  & a & b & c & d & e \\
			2  & b & c & a & d & e \\
			3  & c & a & b & d & e \\
			4  & d & c & e & a & b \\
            5  & e & d & a & b & c
		\end{tabular}
		
		\hspace{1cm}
		
		Worker preferences in $A$ 
		
		\hspace{1cm}
		
	\end{minipage}%
	\begin{minipage}{.45\linewidth}
    \centering
	\begin{tabular}{p{0.5cm}|p{0.5cm}p{0.5cm}p{0.5cm}p{0.5cm}p{0.5cm}}
		a  & 2 & 3 & 1 & 4 & 5 \\
		b  & 3 & 1 & 2 & 4 & 5 \\
		c  & 1 & 2 & 3 & 4 & 5 \\
		d  & 5 & 3 & 4 & 1 & 2 \\
        e  & 4 & 5 & 1 & 2 & 3
	\end{tabular}
	
	\hspace{1cm}
	
	Firm preferences in $A$
	 
	\hspace{1cm}
\end{minipage}%

	\begin{minipage}{.45\linewidth}
        \centering
		\begin{tabular}{p{0.5cm}|p{0.5cm}p{0.5cm}p{0.5cm}p{0.5cm}p{0.5cm}}
		1  & a & b & c & d & e \\
		2  & b & c & a & d & e \\
		\cellcolor{red!25}3  & \cellcolor{red!25}c & \cellcolor{red!25}d & \cellcolor{red!25}a & \cellcolor{red!25}b & \cellcolor{red!25}e \\
		4  & d & c & e & a & b \\
        5  & e & d & a & b & c 
		\end{tabular}

		\hspace{1cm}

		Worker preferences in $B$
	\end{minipage} 
	\begin{minipage}{.45\linewidth}
        \centering
		\begin{tabular}{p{0.5cm}|p{0.5cm}p{0.5cm}p{0.5cm}p{0.5cm}p{0.5cm}}
		a  & 2 & 3 & 1 & 4 & 5 \\
		b  & 3 & 1 & 2 & 4 & 5 \\
		\cellcolor{red!25}c  & \cellcolor{red!25}1 & \cellcolor{red!25}4 & \cellcolor{red!25}3 & \cellcolor{red!25}2 & \cellcolor{red!25}5 \\
		d  & 5 & 3 & 4 & 1 & 2 \\
        e  & 4 & 5 & 1 & 2 & 3
		\end{tabular}

		\hspace{1cm}

		Firm preferences in $B$
	\end{minipage} 
	\end{wbox}
	\caption{ $\mathcal{M}_{A} \setminus \mathcal{M}_B$ is not a semi-sublattice of $\Lc_A$, even when $(A, B)$ is of type $(1,1)$.}
	\label{ex:1-1notSemisublattice} 
\end{figure}



\subsection{Worker and firm optimal stable matchings}
\label{sec:optimal_matchings}

\Cref{thm:sublattice_two_side} states that if the instance $(A, B)$ is of type $(1,n)$, then the set of robust stable matchings forms a sublattice of both $\Lc_A$ and $\Lc_B$. This sublattice structure allows us to define worker-optimal and firm-optimal robust stable matchings. In this section, we provide algorithms to compute these matchings.

We note that for the $(0,n)$ case—i.e., when only one side changes preferences—some algorithms from the literature can be adapted to compute these optimal matchings. We discuss these adaptations in \Cref{app:one_side}. These, however, do not extend to the $(1,n)$ setting, for which we design new algorithms based on Gale and Shapley’s deferred acceptance algorithm~\cite{GaleS}.

The original Deferred Acceptance Algorithm (\Cref{alg:gs_algorithm}) proceeds in iterations. In each iteration, the following steps occur:  
$(a)$ The proposing side (e.g., workers) proposes to their most preferred firm that has not yet rejected them.  
$(b)$ Each firm tentatively accepts its most preferred proposal received in that round and rejects all others.  
$(c)$ Each worker eliminates the firms that rejected them from their preference list.  

The process continues until a perfect matching is formed, at which point the algorithm outputs it as a stable matching. The key idea is that whenever a rejection occurs, the corresponding worker-firm pair can never be part of any stable matching. We modify this algorithm to compute worker- and firm-optimal stable matchings in the intersection $\mathcal{M}_A \cap \mathcal{M}_B$ for the $(1,n)$ setting.

\begin{restatable}{theorem}{nalgoworks}
${}^\dagger$
\label{thm:n_1_algo_works}
Let $A$ and $B$ be two stable matching instances that are of type $(1,n)$. Then \Cref{alg:daalgorithm_worker_both_sides} and \Cref{alg:daalgorithm_firm_both_sides} find the worker- and firm-optimal robust stable matchings, respectively, or correctly report that no such matching exists.
\end{restatable}

\begin{algorithm}[ht]
	\begin{wbox}
		\textsc{WorkerOptimal}($A,B$): \\
		\textbf{Input:} Stable matching instances $A$ and $B$ on agents $W \cup F$.\\
		\textbf{Output:}  Perfect matching $M$ or $\boxtimes$ (when no robust stable matching exists).\\
        Assume there are two rooms, $\mathcal{R}_A$ and $\mathcal{R}_B$ corresponding to instances $A$ and $B$ respectively. Each worker has a list, initialized to all firms, that they look at while proposing.
		\begin{enumerate}
			\item Until there is no rejection in any room or some worker is rejected by all firms, do:
			\begin{enumerate}
				    \item $\forall I \in \{A,B\}, \forall w \in W: w$ proposes to their best-in-$I$ uncrossed firm in $\mathcal{R}_I$.
				    \item $\forall I \in \{A,B\}, \forall f \in F: f$ tentatively accepts their best-in-$I$ proposal in room $\mathcal{R}_I$ and rejects the rest.
				    \item $\forall w \in W:$ if $w$ is rejected by a firm $f$ in \textbf{\textit{any}} room, they cross $f$ off their list. 
			\end{enumerate}  
			\item 
                \begin{enumerate}
                    \item If some worker is rejected by all firms, output $\boxtimes$.
                    \item Else, the acceptances define perfect matchings in each room. If they are the same in all rooms, output the perfect matching ($M$). 
                \end{enumerate}
		\end{enumerate}
	\end{wbox}
	\caption{Algorithm to find the worker-optimal stable matching. Note that proposers (workers) maintain a single list (or set) across rooms. They may propose to different firms in each room but update their list synchronously based on rejections from any room.}
	\label{alg:daalgorithm_worker_both_sides} 
\end{algorithm} 

\begin{algorithm} [ht]
	\begin{wbox}
		\textsc{FirmOptimal}($A,B$): \\
		\textbf{Input:} Stable matching instances $A$ and $B$ on agents $W \cup F$.\\
		\textbf{Output:}  Perfect matching $M$ or $\boxtimes$ (when no robust stable matching exists). \\
        Assume there are two rooms, $\mathcal{R}_A$ and $\mathcal{R}_B$ corresponding to instances $A$ and $B$ respectively. Each firm has a list, initialized to all workers, that they look at while proposing.
		\begin{enumerate}
			\item Until there is no rejection in any room or some firm is rejected by all firms, do:
			\begin{enumerate}
				    \item $\forall I \in \{A,B\}, \forall f \in F: f$ proposes to their best-in-$I$ uncrossed worker in $\mathcal{R}_I$.
				    \item $\forall I \in \{A,B\}, \forall w \in W: w$ tentatively accepts their best-in-$I$ proposal (call it $f^I_w$) in room $\mathcal{R}_I$ and rejects the rest.
                        \item $\forall I \in \{A,B\}, \forall w \in W: w$
                        sends \textbf{\textit{preemptive rejections}} to all firms $f'$ not yet rejected, that are worse than their current option in $\mathcal{R}_I$, i.e., $f^I_w \geq^I_w f'$.
				    \item $\forall f \in F:$ if $f$ is rejected by a worker $w$ in some room, they cross $w$ off their list. 
			\end{enumerate}  
			\item 
                \begin{enumerate}
                    \item If some firm is rejected by all workers, output $\boxtimes$.
                    \item Else, the acceptances define perfect matchings in each room. If they are the same in all rooms, output the perfect matching($M$). 
                \end{enumerate}
		\end{enumerate}

	\end{wbox}
        \caption{Algorithm to find the firm-optimal stable matching. Note that in step 1(c), workers may send rejections to firms that have not yet proposed to them.}
	\label{alg:daalgorithm_firm_both_sides} 
\end{algorithm} 

The worker-optimal algorithm~\ref{alg:daalgorithm_worker_both_sides} runs the DA algorithm simultaneously in multiple rooms, where the preference lists for each room correspond to individual instances. Agents act according to preferences in each room: in every iteration, workers propose and firms tentatively accept their best proposals and reject the rest. However, the rejections apply \emph{across all rooms}. For instance, if a worker is rejected by a firm in room $\mathcal{R}_A$, they remove that firm from their list in \emph{all} rooms. This synchronization ensures that if a perfect matching is returned, it is stable with respect to all instances. The primary technical challenge is the possibility that the algorithm produces distinct perfect matchings in different rooms. \Cref{lem:nomultmatchings} shows this cannot occur, a fact that follows from the intersection $\mathcal{M}_A \cap \mathcal{M}_B$ having the same partial order in both lattices.

The firm-optimal algorithm~\ref{alg:daalgorithm_firm_both_sides} proceeds analogously, with firms proposing and workers tentatively accepting. A key distinction is that workers issue \emph{preemptive} rejections: in each room, a worker rejects \emph{all} firms ranked below their current tentative match—including firms that have not yet proposed—based on their local preference list. This ensures the resulting perfect matching is identical in all rooms.




\subsection{Integrality of the Polytope}
\label{sec:polytope}

\Cref{sec:lp_intro} presented an efficient algorithm for finding stable matchings using linear programming. We now extend this approach to define a linear program for robust stable matchings. Our goal is to characterize when the associated polytope is integral, allowing us to efficiently compute robust stable matchings.

We show that the robust stable matching polytope is integral when $(A, B)$ is of type either $(n,1)$ or $(1,n)$, which in turn yields efficient algorithms for computing robust matchings in these settings.

The linear program for the robust stable matching instance $(A, B)$ is given below. The first two and the last constraints ensure that any feasible solution $x$ is a fractional perfect matching, while the middle two constraints prevent blocking pairs under instances $A$ and $B$, respectively.

\begin{equation}
\label{eq:lp_multi}
    \begin{aligned}[b]
        \text{maximize} \quad & 0 \\
        \text{subject to} \quad
        & \sum_w x_{wf} = 1 && \forall f \in F, \\
        & \sum_f x_{wf} = 1 && \forall w \in W, \\
        & \sum_{f >^A_w f'} x_{wf'} - \sum_{w' >^A_f w} x_{w'f} \leq 0 && \forall w \in W,\ \forall f \in F, \\
        & \sum_{f >^B_w f'} x_{wf'} - \sum_{w' >^B_f w} x_{w'f} \leq 0 && \forall w \in W,\ \forall f \in F, \\
        & x_{wf} \geq 0 && \forall w \in W,\ \forall f \in F.
    \end{aligned}
\end{equation}


A solution to LP~\ref{eq:lp_multi} corresponds to a fractional matching that is stable under both $A$ and $B$. Analogous to the single-instance case (see \Cref{sec:lp_intro}), let $\theta \in [0,1]$ be chosen uniformly at random. Construct two integral perfect matchings $\mu^A_{\theta}$ and $\mu^B_{\theta}$ by performing the interval-based rounding procedure using the preference orders from instances $A$ and $B$, respectively.

\begin{restatable}{lemma}{lp}
${}^\dagger$
\label{lem:lp}
The matchings $\mu^A_{\theta}$ and $\mu^B_{\theta}$ are identical for every $\theta \in [0,1]$, i.e., $\mu^A_{\theta} = \mu^B_{\theta}$.
\end{restatable}

\begin{restatable}{theorem}{nintegralpolytope}
${}^\dagger$
\label{thm:n_1_integral_polytope} 
If $(A,B)$ is of type $(1,n)$ then $\mu^A_{\theta}$ is stable under both $A$ and $B$ and robust fractional stable matching polytope has integer optimal vertices.
\end{restatable}

Hence, in the $(n,1)$ and $(1,n)$ settings, LP~\ref{eq:lp_multi} provides an efficient method to compute a robust stable matching. However, this technique fails for $(2,2)$ instances—the proof of \Cref{lem:lp} does not extend, and, in fact, the polytope may not even be integral, see example in \Cref{ex:2_2_not_integral}.

\begin{restatable}{theorem}{notintegralpolytope}
\label{thm:2_2_not_integral_polytope}
If $(A,B)$ is of type $(2,2)$ then robust stable matching polytope is not always integral.
\end{restatable}

\begin{proof}
Consider the example provided in \Cref{ex:2_2_not_integral}, where instances $A$ and $B$ are defined on 4 workers $\{1,2,3,4\}$ and 4 firms $\{a,b,c,d\}$. Only workers 1 and 2 and firms $a$ and $b$ change their preferences from $A$ to $B$.

Instance $A$ has exactly two stable matchings:
\[
M_A^1 = \{(1,a), (2,b), (3,c), (4,d)\}, \quad
M_A^2 = \{(1,b), (2,a), (3,d), (4,c)\}.
\]
Instance $B$ also has exactly two stable matchings:
\[
M_B^1 = \{(1,b), (2,a), (3,c), (4,d)\}, \quad
M_B^2 = \{(1,a), (2,b), (3,d), (4,c)\}.
\]
And so, there are no robust stable matchings for this pair of $A$ and $B$. 

Now, consider the fractional matching $M = \frac{1}{2}(M_A^1 + M_A^2) = \frac{1}{2}(M_B^1 + M_B^2)$. Since $M$ is a convex combination of stable matchings under both $A$ and $B$, it is a feasible solution to LP~\ref{eq:lp_multi}. Furthermore, $M$ is the unique solution to the LP in this instance.

Therefore, the robust stable matching polytope is a singleton fractional point, showing that the polytope is not integral in general for $(2,2)$ instances.
\end{proof}

\begin{figure}[ht]
	\begin{wbox}
	\centering
	\begin{minipage}{.45\linewidth}
        \centering
		\begin{tabular}{p{0.5cm}|p{0.5cm}p{0.5cm}p{0.5cm}p{0.5cm}p{0.5cm}}
			1  & a & c & b & d \\
			2  & b & a & c & d \\
			3  & c & d & a & b \\
			4  & d & c & a & b \\
		\end{tabular}
		
		\hspace{1cm}
		
		Worker preferences in $A$ 
		
		\hspace{1cm}
		
	\end{minipage}%
	\begin{minipage}{.45\linewidth}
        \centering
	\begin{tabular}{p{0.5cm}|p{0.5cm}p{0.5cm}p{0.5cm}p{0.5cm}p{0.5cm}}
		a  & 2 & 1 & 3 & 4 \\
		b  & 1 & 2 & 3 & 4 \\
		c  & 4 & 1 & 3 & 2 \\
		d  & 3 & 1 & 4 & 2 \\
	\end{tabular}
	
	\hspace{1cm}
	
	Firm preferences in $A$
	 
	\hspace{1cm}
        \end{minipage}%

        \begin{minipage}{.45\linewidth}
        \centering
		\begin{tabular}{p{0.5cm}|p{0.5cm}p{0.5cm}p{0.5cm}p{0.5cm}p{0.5cm}}
			\cellcolor{red!25}1  & \cellcolor{red!25}b & \cellcolor{red!25}d & \cellcolor{red!25}a & \cellcolor{red!25}c \\
			\cellcolor{red!25}2  & \cellcolor{red!25}a & \cellcolor{red!25}b & \cellcolor{red!25}c & \cellcolor{red!25}d \\
			3  & c & d & a & b \\
			4  & d & c & a & b \\
		\end{tabular}
		
		\hspace{1cm}
		
		Worker preferences in $B$ 
		
		\hspace{1cm}
		
	\end{minipage}%
	\begin{minipage}{.45\linewidth}
        \centering
	\begin{tabular}{p{0.5cm}|p{0.5cm}p{0.5cm}p{0.5cm}p{0.5cm}p{0.5cm}}
		\cellcolor{red!25}a  & \cellcolor{red!25}1 & \cellcolor{red!25}2 & \cellcolor{red!25}3 & \cellcolor{red!25}4 \\
		\cellcolor{red!25}b  & \cellcolor{red!25}2 & \cellcolor{red!25}1 & \cellcolor{red!25}3 & \cellcolor{red!25}4 \\
		c  & 4 & 1 & 3 & 2 \\
		d  & 3 & 1 & 4 & 2 \\
	\end{tabular}
	
	\hspace{1cm}
	
	Firm preferences in $B$
	 
	\hspace{1cm}
        \end{minipage}%

	\end{wbox}
	\caption{The robust stable matching polytope for this $(2,2)$ instance is not integral.}
	\label{ex:2_2_not_integral} 
\end{figure}


All the results from this section extend to more than two instances. Let $(A_1, A_2, \ldots, A_k)$ be a robust stable matching instance that is $(1,n)$, i.e., $n-1$ of the $n$ workers have identical preference profiles across all instances. Then:

\begin{restatable}{theorem}{nmultiple}
${}^\dagger$
\label{thm:n_1_multiple}
For a $(1,n)$ robust stable matching instance $(A_1, A_2, \ldots, A_k)$ with $k \geq 2$, with the same worker changing preferences across instances, the set $\calL' = \mathcal{M}_{A_1} \cap \mathcal{M}_{A_2} \cap \ldots \cap \mathcal{M}_{A_k}$ forms a sublattice of each of the lattices $\Lc_{A_i}$. The LP formulation can be used to find the matchings in $\calL'$, and the worker-optimal and firm-optimal robust stable matchings can be computed in polynomial time. Furthermore, if the instances are also $(0,n)$, then Birkhoff's partial order generating $\Lc'$ can be computed efficiently, and its matchings can be enumerated with polynomial delay.
\end{restatable}

\subsection{A Robust Stable Matching Algorithm for the General Case}
\label{sec:robust_general_algo}

Sections~\ref{sec:lattice},~\ref{sec:optimal_matchings}, and~\ref{sec:polytope} establish that finding robust stable matchings for $(1,n)$ instances lies in $\mathcal{P}$, while \cite{NP-Two-stable} shows that the problem is NP-hard when $(A,B)$ is of type $(n,n)$. Consider the general setting when $A$ and $B$ differ arbitrarily on the preferences of agents $S \subseteq W \cup F$, where $|S \cap W| = p$ and $|S \cap F| = q$. \Cref{alg:robust_xp} find a robust stable matching if one exists. By \Cref{thm:robust_xp}, the algorithm runs in time $O(n^{p+q+2})$ and can be modified to enumerate all robust stable matchings with the same asymptotic delay, yielding an XP-time algorithm parameterized by the number of agents whose preferences differ across the two instances. 

Intuitively all possible partner assignments for agents in $S$ are enumerated, there are at most $O(n^{p+q})$ possibilities. Each one defines a partial matching $M$ for the agents in $S$, and the algorithm attempts to determine whether $M$ can be extended to a full matching that is stable under $A$ and $B$. First it verifies that $M$ does not contain any blocking pairs within $T = S \cup M(S)$. If no blocking pair exists, define truncated instance $X$ on the remaining agents $U = (W \cup F) \setminus T$. In $X$, the preference lists of agents in $U$ are shortened to eliminate any potential blocking pairs involving agents in $T$. A stable matching $M'$ is computed on $X$. If $M'$ is a perfect matching on $U$, it can be shown that the combined matching $M^\star = M \cup M'$ is a stable matching on the full set of agents in both $A$ and $B$.

\begin{algorithm}[ht]
  \begin{wbox}
    \textsc{RobustStableMatching}$(A,B)$:\\
    \textbf{Input:} Stable matching instances $A$ and $B$ on agents $W \cup F$.\\
    \textbf{Output:} Perfect matching $M^\star$ or $\boxtimes$ (when no robust stable matching exists).\\
    $S \subseteq W \cup F$ are agents whose preferences differ in $A$ and $B$.
    \begin{enumerate}
      \item \textbf{For} each assignment of partners to agents in $S$ that defines a valid partial matching:
      \begin{enumerate}
        \item Let $M$ be the resulting partial matching and $T \coloneqq S \cup M(S)$.
        \item \textbf{If} there exists $w \in T \cap W$ and $f \in T \cap F$ such that $(w,f)$ is a blocking pair with respect to $M$ in either $A$ or $B$, \textbf{continue} to the next iteration.
        \item Remove all agents in $T$ to obtain the instance $X$ on $U \coloneqq (W \cup F)\setminus T$;
        \item For each $a \in T$ and $b \in U$, if $b \succ_a M(a)$ in either $A$ or $B$, then truncate $b$'s preference list at $a$ in $X$ (i.e., remove $a$ and all agents ranked below $a$).
        \item Find stable matching $M'$ in $X$.
        \item {If} $M'$ is a perfect matching on $U$, \textbf{return} $M^\star \coloneqq M \cup M'$.
      \end{enumerate}
      \item \textbf{Return} $\boxtimes$.
    \end{enumerate}
  \end{wbox}
  \caption{An XP-time algorithm to find a robust stable matching for small $(p,q)$.}
  \label{alg:robust_xp}
\end{algorithm}

\begin{restatable}{theorem}{robustalgoworks}
${}^\dagger$
\label{thm:robust_xp}
    For instances $A$ and $B$ of type $(1,n)$. \Cref{alg:robust_xp} finds a robust stable matching in $O\!\left(n^{p+q+2}\right)$ time, or correctly reports that no such matching exists. All the set of robust stable matchings can be enumerated with $O\!\left(n^{p+q+2}\right)$ delay.
 \end{restatable}
\begin{restatable}{corollary}{robust_xp}
\label{cor:xp_time}
If a constant number of agents change their preferences, then deciding whether a robust stable matching exists and finding one can be achieved in polynomial time. All such matchings can be enumerated with polynomial delay.
\end{restatable}
\Cref{thm:robust_xp} implies \Cref{cor:xp_time} as the running time is polynomial if a constant number of agents change preferences. The decision and enumeration problems for robust stable matchings can be solved efficiently. This shows that the set of robust stable matchings can be expressed as a union of $O(n^{p+q})$ stable matching lattices, one for each consistent assignment of partners to agents in $S$. 

However the partial orders defining these lattices can differ significantly. It remains unclear whether this decomposition yields an efficient method for computing Birkhoff’s partial order even in the $(1,1)$ case—and by \Cref{cor:birkhoff_one_n}, also in the $(1,n)$ case. These observations naturally lead to the following open problems.


\begin{remark}
\textbf{Open Problem.} The problem of computing Birkhoff’s partial order for the $(1,1)$ case and by extension, for $(n,1)$ is open.
\end{remark}

\begin{remark}
\textbf{Open Problem.} Given a $(p,q)$ robust stable matching instance $(A,B)$ with $2 \leq p,q \leq n$, $p+q < 2n$, and $p+q = \omega(1)$, it remains open whether determining the existence of a robust stable matching—and computing one if it exists—is in $\mathcal{P}$.
\end{remark}

\section*{Acknowledgments}
We thank all the reviewers for their valuable comments. We are especially grateful to Simon Murray for his insights, which led to~\Cref{alg:robust_xp}.

\bibliographystyle{alpha}
\bibliography{refs}

\newcommand{\etalchar}[1]{$^{#1}$}
\begin{thebibliography}{GMRV22}

\bibitem[ABF{\etalchar{+}}17]{aziz2}
H.~Aziz, P.~Biro, T.~Fleiner, S.~Gaspers, R.~de~Haan, N.~Mattei, and B.~Rastegari.
\newblock Stable matching with uncertain pairwise preferences.
\newblock In {\em International Joint Conference on Autonomous Agents and Multiagent Systems}, pages 344--352, 2017.

\bibitem[ABG{\etalchar{+}}16]{aziz1}
H.~Aziz, P.~Biro, S.~Gaspers, R.~de~Haan, N.~Mattei, and B.~Rastegari.
\newblock Stable matching with uncertain linear preferences.
\newblock In {\em International Symposium on Algorithmic Game Theory}, pages 195--206, 2016.

\bibitem[BBHN21]{BBHN21a}
Niclas Boehmer, Robert Bredereck, Klaus Heeger, and Rolf Niedermeier.
\newblock Bribery and control in stable marriage.
\newblock {\em Journal of Artificial Intelligence Research}, 71:993--1048, 2021.

\bibitem[BCK{\etalchar{+}}20]{BredereckCKLN20}
Robert Bredereck, Jiehua Chen, Dusan Knop, Junjie Luo, and Rolf Niedermeier.
\newblock Adapting stable matchings to evolving preferences.
\newblock In {\em The Thirty-Fourth {AAAI} Conference on Artificial Intelligence, {AAAI} 2020, The Thirty-Second Innovative Applications of Artificial Intelligence Conference, {IAAI} 2020, The Tenth {AAAI} Symposium on Educational Advances in Artificial Intelligence, {EAAI} 2020, New York, NY, USA, February 7-12, 2020}, pages 1830--1837. {AAAI} Press, 2020.

\bibitem[BCK24]{BERCZI}
Kristóf Bérczi, Gergely Csáji, and Tamás Király.
\newblock Manipulating the outcome of stable marriage and roommates problems.
\newblock {\em Games and Economic Behavior}, 147:407--428, 2024.

\bibitem[BFK{\etalchar{+}}21]{BFK+21a}
Robert Bredereck, Piotr Faliszewski, Andrzej Kaczmarczyk, Rolf Niedermeier, Piotr Skowron, and Nimrod Talmon.
\newblock Robustness among multiwinner voting rules.
\newblock {\em Artificial Intelligence}, 290:103403, 2021.

\bibitem[BGS24]{robust_pop_matchings_BGS}
Martin Bullinger, Rohith~Reddy Gangam, and Parnian Shahkar.
\newblock Robust popular matchings.
\newblock In {\em Proceedings of the 23rd International Conference on Autonomous Agents and Multiagent Systems}, AAMAS '24, page 225–233, Richland, SC, 2024. International Foundation for Autonomous Agents and Multiagent Systems.

\bibitem[BHN22a]{boehmer_incremental_mfcs}
Niclas Boehmer, Klaus Heeger, and Rolf Niedermeier.
\newblock {Deepening the (Parameterized) Complexity Analysis of Incremental Stable Matching Problems}.
\newblock In Stefan Szeider, Robert Ganian, and Alexandra Silva, editors, {\em 47th International Symposium on Mathematical Foundations of Computer Science (MFCS 2022)}, volume 241 of {\em Leibniz International Proceedings in Informatics (LIPIcs)}, pages 21:1--21:15, Dagstuhl, Germany, 2022. Schloss Dagstuhl -- Leibniz-Zentrum f{\"u}r Informatik.

\bibitem[BHN22b]{boehmer_incremental_aaai}
Niclas Boehmer, Klaus Heeger, and Rolf Niedermeier.
\newblock Theory of and experiments on minimally invasive stability preservation in changing two-sided matching markets.
\newblock {\em Proceedings of the AAAI Conference on Artificial Intelligence}, 36(5):4851--4858, Jun. 2022.

\bibitem[Bir37]{Birkhoff}
Garrett Birkhoff.
\newblock Rings of sets.
\newblock {\em Duke Mathematical Journal}, 3(3):443--454, 1937.

\bibitem[CNS18]{NP-Two-stable-Chen}
Jiehua Chen, Rolf Niedermeier, and Piotr Skowron.
\newblock Stable marriage with multi-modal preferences.
\newblock In {\em Proceedings of the 2018 ACM Conference on Economics and Computation}, EC '18, page 269–286, New York, NY, USA, 2018. Association for Computing Machinery.

\bibitem[Cs{\'a}24]{Csaj24a}
Gergely Cs{\'a}ji.
\newblock Popular and dominant matchings with uncertain and multimodal preferences.
\newblock In Kate Larson, editor, {\em Proceedings of the Thirty-Third International Joint Conference on Artificial Intelligence, {IJCAI-24}}, pages 2740--2747. International Joint Conferences on Artificial Intelligence Organization, 8 2024.
\newblock Main Track.

\bibitem[CSS21]{Chen-Matchings-Under-preferences}
Jiehua Chen, Piotr Skowron, and Manuel Sorge.
\newblock Matchings under preferences: Strength of stability and tradeoffs.
\newblock {\em ACM Trans. Econ. Comput.}, 9(4), oct 2021.

\bibitem[DF81]{DubinsF}
Lester~E Dubins and David~A Freedman.
\newblock Machiavelli and the {G}ale-{S}hapley algorithm.
\newblock {\em The American Mathematical Monthly}, 88(7):485--494, 1981.

\bibitem[EFS09]{EFS09a}
Edith Elkind, Piotr Faliszewski, and Arkadii Slinko.
\newblock On distance rationalizability of some voting rules.
\newblock In {\em Proceedings of the 12th Conference on Theoretical Aspects of Rationality and Knowledge}, TARK '09, page 108–117, New York, NY, USA, 2009. Association for Computing Machinery.

\bibitem[EIV23]{Book-Online}
Federico Echenique, Nicole Immorlica, and Vijay~V. Vazirani, editors.
\newblock {\em Online and Matching-Based Market Design}.
\newblock Cambridge University Press, 2023.

\bibitem[FIM07]{Fleiner}
Tamás Fleiner, Robert~W. Irving, and David~F. Manlove.
\newblock Efficient algorithms for generalized stable marriage and roommates problems.
\newblock {\em Theoretical Computer Science}, 381(1):162--176, 2007.

\bibitem[FKJ16]{genPreferences}
Linda Farczadi, Georgiou Konstantinos, and Könemann Jochen.
\newblock Stable marriage with general preferences.
\newblock {\em Theory of Computing Systems}, 59(4):683--699, 2016.

\bibitem[FR16]{FaRo15a}
Piotr Faliszewski and J{\"o}rg Rothe.
\newblock Control and bribery in voting.
\newblock In Felix Brandt, Vincent Conitzer, Ulle Endriss, J.~Lang, and Ariel~D. Procaccia, editors, {\em Handbook of Computational Social Choice}, chapter~7. Cambridge University Press, 2016.

\bibitem[G{\"a}r75]{Gard75a}
Peter G{\"a}rdenfors.
\newblock Match making: {A}ssignments based on bilateral preferences.
\newblock {\em Behavioral Science}, 20(3):166--173, 1975.

\bibitem[GI89]{GusfieldI}
Dan Gusfield and Robert~W Irving.
\newblock {\em The stable marriage problem: structure and algorithms}.
\newblock MIT press, 1989.

\bibitem[GMRV22]{GMRV-nearby-instances}
Rohith~Reddy Gangam, Tung Mai, Nitya Raju, and Vijay~V. Vazirani.
\newblock {A Structural and Algorithmic Study of Stable Matching Lattices of "Nearby" Instances, with Applications}.
\newblock In {\em 42nd IARCS Annual Conference on Foundations of Software Technology and Theoretical Computer Science (FSTTCS 2022)}, 2022.

\bibitem[GS62]{GaleS}
David Gale and Lloyd~S Shapley.
\newblock College admissions and the stability of marriage.
\newblock {\em The American Mathematical Monthly}, 69(1):9--15, 1962.

\bibitem[GSOS17]{genc2}
Begum Genc, Mohamed Siala, Barry O'Sullivan, and Gilles Simonin.
\newblock Finding robust solutions to stable marriage.
\newblock {\em arXiv preprint arXiv:1705.09218}, 2017.

\bibitem[GSSO17]{genc1}
Begum Genc, Mohamed Siala, Gilles Simonin, and Barry O’Sullivan.
\newblock On the complexity of robust stable marriage.
\newblock In {\em International Conference on Combinatorial Optimization and Applications}, pages 441--448. Springer, 2017.

\bibitem[IL86]{irving2}
Robert~W Irving and Paul Leather.
\newblock The complexity of counting stable marriages.
\newblock {\em SIAM Journal on Computing}, 15(3):655--667, 1986.

\bibitem[Irv94]{IRVING-Indifferences}
Robert~W. Irving.
\newblock Stable marriage and indifference.
\newblock {\em Discrete Applied Mathematics}, 48(3):261--272, 1994.

\bibitem[Knu76]{knuth1976marriages}
Donald~Ervin Knuth.
\newblock Marriages stables.
\newblock {\em Technical report}, 1976.

\bibitem[Knu97]{Knuth}
Donald~Ervin Knuth.
\newblock {\em Stable marriage and its relation to other combinatorial problems: An introduction to the mathematical analysis of algorithms}, volume~10.
\newblock American Mathematical Soc., 1997.

\bibitem[KPG16]{Kunysz-ssm}
Adam Kunysz, Katarzyna~E. Paluch, and Pratik Ghosal.
\newblock Characterisation of strongly stable matchings.
\newblock In Robert Krauthgamer, editor, {\em Proceedings of the Twenty-Seventh Annual {ACM-SIAM} Symposium on Discrete Algorithms, {SODA} 2016, Arlington, VA, USA, January 10-12, 2016}, pages 107--119. {SIAM}, 2016.

\bibitem[Man02]{MANLOVE-indifference-lattice}
David~F. Manlove.
\newblock The structure of stable marriage with indifference.
\newblock {\em Discrete Applied Mathematics}, 122(1):167--181, 2002.

\bibitem[Man13]{Manlove-book}
David Manlove.
\newblock {\em Algorithmics of Matching Under Preferences}.
\newblock World Scientific, 2013.

\bibitem[ML18]{Menon_Larson}
Vijay Menon and Kate Larson.
\newblock Robust and approximately stable marriages under partial information.
\newblock In George Christodoulou and Tobias Harks, editors, {\em Web and Internet Economics}, pages 341--355, Cham, 2018. Springer International Publishing.

\bibitem[MO19]{NP-Two-stable}
Shuichi Miyazaki and Kazuya Okamoto.
\newblock Jointly stable matchings.
\newblock {\em J. Comb. Optim.}, 38(2):646–665, 2019.

\bibitem[MV18]{MV.robust}
Tung Mai and Vijay~V. Vazirani.
\newblock Finding stable matchings that are robust to errors in the input.
\newblock In {\em European Symposium on Algorithms}, 2018.

\bibitem[Rot82]{Roth-IC-1982economics}
Alvin~E Roth.
\newblock The economics of matching: Stability and incentives.
\newblock {\em Mathematics of operations research}, 7(4):617--628, 1982.

\bibitem[Rot85]{Roth85}
A.~E. Roth.
\newblock The college admissions problem is not equivalent to the marriage problem.
\newblock {\em Journal of Economic Theory}, 36(2):277--288, 1985.

\bibitem[Spi95]{SPIEKER-indifference-lattice}
Boris Spieker.
\newblock The set of super-stable marriages forms a distributive lattice.
\newblock {\em Discrete Applied Mathematics}, 58(1):79--84, 1995.

\bibitem[Sta96]{Stanley}
Richard Stanley.
\newblock Enumerative combinatorics, vol. 1, {W}adsworth and {B}rooks/{C}ole, {P}acific {G}rove, {C}{A}, 1986; second printing, 1996.

\bibitem[SYE13]{SYE13a}
Dmitry Shiryaev, Lan Yu, and Edith Elkind.
\newblock On elections with robust winners.
\newblock In {\em Proceedings of the 2013 International Conference on Autonomous Agents and Multi-Agent Systems}, AAMAS '13, page 415–422, Richland, SC, 2013. International Foundation for Autonomous Agents and Multiagent Systems.

\bibitem[TS98]{Teo-Sethuraman-Lp}
Chung-Piaw Teo and Jay Sethuraman.
\newblock The geometry of fractional stable matchings and its applications.
\newblock {\em Mathematics of Operations Research}, 23(4):874--891, 1998.

\bibitem[{Van}89]{VANDEVATE}
John~H. {Vande Vate}.
\newblock Linear programming brings marital bliss.
\newblock {\em Operations Research Letters}, 8(3):147--153, 1989.

\end{thebibliography}

\appendix

\section{The Partial Order of Rotations}
\label{app:rotations} 

For the lattice of stable matchings, the partial order $\Pi$ defined in Birkhoff's Theorem has additional useful structural properties. First, its elements are rotations. 
A {\em rotation} takes $r$ matched worker-firm pairs in a fixed order, say 
$\{w_0f_0, w_1f_1,\ldots, w_{r-1}f_{r-1}\}$, and cyclically changes the matches of these $2r$ 
agents. The number $r$, the $r$ pairs, and their order are chosen so that when a rotation is applied to a stable matching containing all $r$ pairs, the resulting matching is also stable; moreover, there is no valid rotation on any subset of these $r$ pairs, under
any ordering. Hence, a rotation can be viewed as a minimal
change to the current matching that results in a stable matching.

Any worker-firm pair $(w, f)$ belongs to at most one rotation. Consequently, the set $R$ of rotations underlying $\Pi$ satisfies $|R| = O(n^2)$, and hence, $\Pi$ is a succinct representation of $\Lc$ whereas the latter can be exponentially large. $\Pi$ will be called the {\em rotation poset} for $\Lc$.

Second, the rotation poset helps traverse the lattice as follows.
For any closed set $S$ of $\Pi$, the corresponding stable matching $M(S)$ can be obtained
as follows: start from the worker-optimal matching in the lattice and apply
the rotations in set $S$, in any topological order consistent with $\Pi$. The resulting
matching will be $M(S)$. In particular, applying all rotations in $R$, starting from the worker-optimal
matching, leads to the firm-optimal matching.

The following process yields a rotation for a stable matching $M$. For a worker $w$ let $s_M(w)$ denote the first firm $f$ on $w$'s list such that $f$ strictly prefers $w$ to its $M$-partner. Let $next_M(w)$ denote the partner in $M$ of firm $s_M(w)$. A \emph{rotation} $\rho$ \emph{exposed} in $M$ is an ordered list of pairs $\{w_0f_0, w_1f_1,\ldots, w_{r-1}f_{r-1}\}$ such that for each $i$, $0 \leq i \leq r-1$, $w_{i+1}$ is $next_M(w_i)$, where $i+1$ is taken modulo $r$. In this paper, we assume that the subscript is taken modulo $r$ whenever we mention a rotation. Notice that a rotation is cyclic and the sequence of pairs can be cyclically rotated. $M / \rho$ is defined to be a matching in which each worker not in a pair of $\rho$ stays matched to the same firm and each worker $w_i$ in $\rho$ is matched to $f_{i+1} = s_M(w_i)$. It can be proven that $M / \rho$ is also a stable matching. The transformation from $M$ to $M / \rho$ is called the \emph{elimination} of $\rho$ from $M$.

\begin{lemma}[\cite{GusfieldI}, Theorem 2.5.4]
	\label{lem:seqElimination}
	Every rotation appears exactly once in any sequence of elimination from $M_0$ to $M_z$.
\end{lemma}

Let $\rho = \{w_0f_0, w_1f_1,\ldots, w_{r-1}f_{r-1}\}$ be a rotation. For $0 \leq i \leq r-1$, we say that $\rho$ {moves $w_i$ from $f_i$ to $f_{i+1}$}, and {moves $f_i$ from $w_{i}$ to $w_{i-1}$}. If $f$ is either $f_i$ or is strictly between $f_{i}$ and $f_{i+1}$ in $w_i$'s list, then we say that $\rho$ {moves $w_i$ below $f$}. Similarly, $\rho$ {moves $f_i$ above} $w$ if $w$ is $w_i$ or between $w_i$ and $w_{i-1}$ in $f_i$'s list.

\subsection{The rotation poset}

A rotation $\rho'$ is said to \emph{precede} another rotation $\rho$, denoted by $\rho' \prec \rho$, if $\rho'$ is eliminated in every sequence of eliminations from $M_0$ to a stable matching in which $\rho$ is exposed. 
If $\rho'$ precedes $\rho$, we also say that $\rho$ \emph{succeeds} $\rho'$.
If neither $\rho' \prec \rho$ nor $\rho' \succ \rho$, we say that $\rho'$ and $\rho$ are 
\emph{incomparable}.
Thus, the set of rotations forms a partial order via
this precedence relationship. The partial order on rotations is called \emph{rotation poset} and denoted by $\Pi$.

\begin{lemma}[\cite{GusfieldI}, Lemma 3.2.1]
	\label{lem:pre2}
	For any worker $w$ and firm $f$, there is at most one rotation that moves $w$ to $f$, $w$ below $f$, or $f$ above $w$. Moreover, if $\rho_1$ moves $w$ to $f$ and $\rho_2$ moves $w$ from $f$ then $\rho_1 \prec \rho_2$.
\end{lemma}

\begin{lemma}[\cite{GusfieldI}, Lemma 3.3.2]
	\label{lem:computePoset}
	$\Pi$ contains at most $O(n^2)$ rotations and can be computed in polynomial time.
\end{lemma}

Consequently, $\Pi$ is a succinct representation of $\Lc$; the latter can be exponentially large. Rotations also help in the efficient enumeration of stable matchings (\cite{GusfieldI}, Chapter 3.5). If we know the edge set defining a sublattice (see \Cref{sec.birkhoff} for the definition) $\calL'$ of a stable matching lattice, the corresponding compression defines a partial order on the ``meta-rotations''(groups of rotations). This new partial order can be used to efficiently enumerate the matchings in sublattice $\calL'$.


\section{Omitted Proofs}
\label{app:theorem_proofs}
\lemnnotsemisublattice*

\begin{proof}[Proof of \Cref{lem:0_n_not_semi_sublattice}]
Consider instances $A$ and $B$ as given in \Cref{ex:oneSideNotSemiSubLattice}. The set of workers is $\mathcal{W} = \{1, 2, 3, 4, 5\}$ and the set of firms is $\mathcal{F} = \{a, b, c, d, e\}$. Instance $B$ is obtained from instance $A$ by permuting the lists of firms $b$ and $c$. Consider matchings $M_1 = \{1b, 2a, 3c, 4d, 5e\}$ and $M_2 = \{1a, 2b, 3d, 4c, 5e\}$ that are both stable matchings with respect to instance $A$. With respect to instance $B$, $M_1$ has the blocking pair $(5, c)$ and $M_2$ has the blocking pair $(4, b)$. Hence, $M_1$ and $M_2$ are both in $\mathcal{M}_{A} \setminus \mathcal{M}_B$. However, the \emph{join} with respect to $A$, $M_1 \vee_A M_2 = \{1b, 2a, 3d, 4c, 5e\}$, is stable with respect to both $A$ and $B$.

Similarly, we see that $M_3 = \{1b, 2c, 3d, 4a, 5e\}$ and $M_4 = \{1d, 2a, 3b, 4c, 5e\}$ are stable with respect to $A$, but have blocking pairs $(1, a)$ and $(4, d)$ respectively, under $B$. However, their \emph{meet}, $M_3 \wedge_A M_4 = \{1b, 2a, 3d, 4c, 5e\}$, is stable with respect to both $A$ and $B$.

Hence, $\mathcal{M}_{A} \setminus \mathcal{M}_B$ is not a semi-sublattice of $\Lc_A$.
\end{proof}

\bothsidesnotsublattice*

\begin{proof} [Proof of \Cref{thm:both_sides_not_sublattice}]
The example in \Cref{ex:twoSidesNotSublattice} shows that the intersection is not necessarily a sublattice in $(2, 2)$. Instance $B$ differs from $A$ in the preferences of workers $3$ and $4$ and firms $c$ and $d$.

Consider matchings $M_1 = \{1a, 2b, 3c, 4d\}$ and $M_2 = \{1b, 2a, 3d, 4c\}$ that are stable under $A$ and $B$. However, $M = M_1 \wedge_A M_2 = \{1a, 2b, 3d, 4c\}$ is not stable under $B$, as $a >_4 c$ and $4 >_a 1$, so $4a$ is a blocking pair of $M$ under $B$. 

Therefore, if there are at least two agents on each side having different preferences in $A$ and $B$, $\mathcal{M}_A \intersect \mathcal{M}_B$ is not a sublattice of $\Lc_B$. 
\end{proof}

\breakinginstancesoneside*

\begin{proof}[Proof of \Cref{thm:breaking_instances_one_side}]
We will show that any blocking pair for a perfect matching under the preference lists of instance $A$ or instance $B$ is a blocking pair under the preference lists of at least one of the instances $A$ or $B_1, \dots, B_n$, and vice versa.

Let $M$ be a perfect matching with a blocking pair under the preference lists in $A$ or in $B$. If it has a blocking pair under $A$, then $M$ is clearly not in $\mathcal{M}_{A} \intersect \mathcal{M}_{B_1} \intersect \ldots \intersect \mathcal{M}_{B_n}$. Suppose $(w, f)$ is a blocking pair for $B$ and $f >_w^B M(w)$, and so $w >_f^B M(f)$. Workers have identical lists in all instances, so $w$ prefers $f$ to $M(w)$ in every instance $B_i$. Firm $f$ has the same preference list in $B$ and in $B_k$ for some $k \in \{1, \dots, n\}$. Then $w >_f^{B_k} M(f)$, and $(w, f)$ is a blocking pair for instance $B_k$, so $M$ is not a stable matching in $\mathcal{M}_{A} \intersect \mathcal{M}_{B_1} \intersect \ldots \intersect \mathcal{M}_{B_n}$.

Let $M$ be a perfect matching with blocking pair $(w, f)$ under the preference lists in one of the instances $A, B_1, \dots, B_n$. As in the previous case, if $(w, f)$ is a blocking pair under the preference lists of $A$, then $M$ is not a stable matching in $\mathcal{M}_A \intersect \mathcal{M}_B$. Otherwise, $(w, f)$ is a blocking pair for some instance $B_k$. We will show that $(w, f)$ must be a blocking pair for $B$. 

Assume that $f$ is not the agent whose preference list is identical in $B$ and $B_k$. Then both $w$ and $f$ have identical preference lists in $A$ and in $B_k$, and so this is a blocking pair for $A$. Otherwise, it must be the case that $f$ changed its preference list in $B_k$ and has the same changed list in $B$. Then $(w, f)$ is a blocking pair for $B$ as well, as $f >_w^{B} M(w)$ and $w >_f^{B} M(f)$. We conclude that the two sets are equal. 
\end{proof}

\breakinginstancestwoside*
\begin{proof}[Proof of \Cref{thm:breaking_instances_two_side}]
    $(\supseteq)$
    Let $\mu \notin \mathcal{M}_A \cap \mathcal{M}_B$. Then $\mu$ has a blocking pair with respect to $A$ or $B$. If it is $A$, then it is trivially not in $\mathcal{M}_A \cap \mathcal{M}_{B_{\{w_1, f_1\}}} \cap \mathcal{M}_{B_{\{w_1, f_2\}}} \cap \dots \cap \mathcal{M}_{B_{\{w_1, f_n\}}}$. So suppose it has a blocking pair $(w, f)$ with respect to $B$. 
    
    Suppose $w = w_1$ and $f \in \{f_1, \dots, f_n\}$. So $f >_{w}^B \mu(w)$ and $w >_{f}^B \mu(f)$. But these are the preferences of $w$ and $f$ in $B_{\{w, f\}}$, so $\mu \notin \mathcal{M}_A \cap \mathcal{M}_{B_{\{w_1, f_1\}}} \cap \mathcal{M}_{B_{\{w_1, f_2\}}} \cap \dots \cap \mathcal{M}_{B_{\{w_1, f_n\}}}$. 

    Else, if $w \neq w_1$ and $f \in \{f_1, \dots, f_n\}$, then $w$ has the same preferences in $A$ and $B$, as well as in every $B_{\{w_1, f\}}$, while $f$ has the same preferences in $B$ and $B_{\{w_1, f\}}$. Hence, $(w, f)$ is a blocking pair for instance $B_{w_1, f}$. 
    
    $(\subseteq)$
    Similarly, consider $\mu \notin \mathcal{M}_A \cap \mathcal{M}_{B_{w_1, f_1}} \cap \dots \cap \mathcal{M}_{B_{w_1, f_n}}$. Then $\mu$ has a blocking pair $(w, f)$, and as in the previous case, assume the blocking pair is not with respect to $A$. Then it is with respect to some instance $B_{\{w_1, f_i\}}$. 

    There are three cases:
    \begin{enumerate} 
        \item $w \neq w_1$ and $f = f_i$: Then $f_i$ has the same list in $B$ and $B_{w_1, f_i}$, $w$ has the same list in $A$, $B$, and $B_{w_1, f_i}$; hence it is a blocking pair with respect to $B$.
        
        \item $w = w_1$ and $f \neq f_i$: Then $w_1$ has the same list in $A$, $B$, and $B_{\{w_1, f_i\}}$, and $f$ has the same list in $A$, $B$, and $B_{\{w_1, f_i\}}$; hence it is a blocking pair with respect to $B$.
        
        \item $w = w_1$ and $f = f_i$: Then $w_1$ has the same list in $B$ and $B_{\{w_1, f_i\}}$, and $f$ has the same list in $B$ and $B_{\{w_1, f_i\}}$; hence it is a blocking pair with respect to $B$.
    \end{enumerate}
    The other case of $w \neq w_1$ and $f \neq f_i$ does not arise because then $(w,f)$ would be a blocking pair of $A$, which cannot be, as $\mu \in \mathcal{M}_A$. 
    
    Hence, $(w, f)$ is a blocking pair with respect to $B$, and $\mu \notin \mathcal{M}_B$; hence it is not in $\mathcal{M}_A \cap \mathcal{M}_B$.

    We conclude that the sets are the same. 
\end{proof}

\lnotsemisublattice*
\begin{proof}[Proof of \Cref{lem:1_1_not_semi_sublattice}]
Consider instances $A$ and $B$ as given in \Cref{ex:1-1notSemisublattice}. The set of workers is $\mathcal{W} = \{1, 2, 3, 4, 5\}$ and the set of firms is $\mathcal{F} = \{a, b, c, d, e\}$. Instance $B$ is obtained from instance $A$ by permuting the lists of firm $c$ and worker $3$. Consider matchings $M_1 = \{1b, 2c, 3a, 4d, 5e\}$ and $M_2 = \{1a, 2b, 3c, 4e, 5d\}$ that are both stable with respect to $A$. With respect to instance $B$, $M_1$ has the blocking pair $(3, d)$ and $M_2$ has the blocking pair $(4, c)$. Hence, $M_1$ and $M_2$ are both in $\mathcal{M}_{A} \setminus \mathcal{M}_B$. However, the \emph{join} with respect to $A$, $M_1 \vee_A M_2 = \{1a, 2b, 3c, 4d, 5e\}$, is stable with respect to $A$ and $B$. 

Similarly, we see that $M_3 = \{1c, 2a, 3b, 4d, 5e\}$ and $M_4 = \{1b, 2c, 3a, 4e, 5d\}$ are stable with respect to $A$ but have blocking pairs $(3, d)$ and $(4, c)$ respectively, under $B$. However, their \emph{meet}, $M_3 \wedge_A M_4 = \{1c, 2a, 3b, 4e, 5d\}$, is stable with respect to $A$ and $B$. 

Hence, $\mathcal{M}_{A} \setminus \mathcal{M}_B$ is not a semi-sublattice of $\Lc_A$. 
\end{proof}

\lp*
\begin{proof} [Proof of \Cref{lem:lp}]
    $\mu^A_{\theta}$ and $\mu^B_{\theta}$ are perfect matchings, and we know that the preferences of $n-1$ workers $w_1, \dots, w_{n-1}$ are the same in $A$ and $B$. This means that these $n-1$ workers are paired to $n-1$ distinct firms in both matchings. Since they are perfect matchings, the $n^{\text{th}}$ worker must be paired with the remaining firm. Hence, the matchings are identical. 
\end{proof}
\nintegralpolytope*

\begin{proof} [Proof of \Cref{thm:n_1_integral_polytope}]
    From \Cref{lem:lp}, $\mu^A_{\theta}$ and $\mu^B_{\theta}$ are the same. Therefore, since it is integral and a solution to LP~\ref{eq:lp_multi}, it is stable under both $A$ and $B$. This means that the set of all fractional robust stable matchings can be written as a convex combination of integral solutions, similar to \Cref{thm:lp_polytope}. This proves that the polytope is integral.
\end{proof}

\nmultiple*

\begin{proof}[Proof of \Cref{thm:n_1_multiple}]
    $\calL'$ is a sublattice of each $\calL_{A_i}$ as a corollary of \Cref{thm:sublattice_two_side}. The LP approach can be used to find stable matchings in the intersection. The new formulation will have $k$ blocking pair conditions—one for each instance $A_i$—to ensure that no worker–firm pair $(w,f)$ is a blocking pair for $x$ under that instance. \Cref{alg:daalgorithm_worker_both_sides} and \Cref{alg:daalgorithm_firm_both_sides} can be used to find the worker- and firm-optimal matchings. Define a room for each instance and run the Deferred Acceptance algorithm in each room according to the preferences of the corresponding instance. If the same perfect matching $M$ is obtained in every room, it is output by the algorithm and is stable with respect to all instances.

    When the instances are $(0,n)$, Birkhoff's partial order for $\calL'$ can be found by identifying the edge set representing $\mathcal{M}_{A_1} \cap \mathcal{M}_{A_i}$ for each instance $A_i$, and by \Cref{lem:1_0_multiple}, the union of these edge sets defines $\calL'$. This provides a way to efficiently enumerate the matchings.
\end{proof}

\section{Proofs of Correctness of Algorithms~\ref{alg:daalgorithm_worker_both_sides},~\ref{alg:daalgorithm_firm_both_sides} and~\ref{alg:robust_xp}}
\label{app:algo_proofs}

\nalgoworks*

The proof of \Cref{thm:n_1_algo_works} is split between \Cref{sec:daalgorithm_worker_both_sides} and \Cref{sec:daalgorithm_firm_both_sides}.

\subsection{Proof of Correctness of Algorithm~\ref{alg:daalgorithm_worker_both_sides}}
\label{sec:daalgorithm_worker_both_sides}

The algorithm involves running the Deferred Acceptance algorithm in different rooms with workers proposing, and returning the worker-optimal matching in $\mathcal{M}_A \cap \mathcal{M}_B$, given that $A$ and $B$ are $(1,n)$. Let $w_1$ be the worker with different preferences in $A$ and $B$.
\begin{lemma}
\label{lem:nomultmatchings}
    \Cref{alg:daalgorithm_worker_both_sides} terminates and always returns a perfect matching or $\boxtimes$.
\end{lemma}

\begin{proof}[Proof of \Cref{lem:nomultmatchings}]
    The algorithm always terminates since there is at least one rejection per iteration before the final iteration. 

    Suppose the algorithm terminates with two distinct perfect matchings, $\mu_A$ in room $A$ and $\mu_B$ in room $B$. 
    Then there must exist a worker $w \neq w_1$ with the same preferences in $A$ and $B$ but different matches in $\mu_A$ and $\mu_B$. Such a worker must exist since only $w_1$ changed their preferences, and the matchings are perfect. Suppose, without loss of generality, that $\mu_A(w) >_w^{A, B} \mu_B(w)$. As $w$ has the same lists in $A$ and $B$, $w$ must have proposed to $\mu_A(w)$ in room $B$, been rejected, and crossed them off their list. However, per the algorithm, $w$ should have also crossed off $\mu_A(w)$ from their list in room $A$, hence $\mu_A$ cannot be a perfect matching achieved in room $A$. 
\end{proof}

\begin{lemma}
\label{lem:n_1_worker_bp}
    If firm $f$ rejects $w$ in either room, $w$ and $f$ cannot be matched in any matching in $\mathcal{M}_A \cap \mathcal{M}_B$.
\end{lemma}

\begin{proof}[Proof of \Cref{lem:n_1_worker_bp}]
    We use the property of the DA algorithm for a single instance $A$ which states that if a firm $f$ rejects a worker $w$, then $(w, f)$ are not matched in any stable matching with respect to $A$~\cite{GaleS}. 

    Suppose, without loss of generality, that $f$ rejects $w$ in room $A$. Then $(w, f)$ are not matched in any stable matching in $\mathcal{M}_A$, and hence they cannot be matched in any stable matching in $\mathcal{M}_A \cap \mathcal{M}_B$.
\end{proof}

\begin{lemma}
\label{lem:daalgorithm_worker_both_sides}
    If $A$ and $B$ admit a stable matching under both, then \Cref{alg:daalgorithm_worker_both_sides} finds one.
\end{lemma}

\begin{proof}[Proof of \Cref{lem:daalgorithm_worker_both_sides}]
Suppose there exists such a stable matching, but the algorithm terminates by outputting $\boxtimes$. This only happens when a worker is rejected by all firms. However, by \Cref{lem:n_1_worker_bp}, this means that the worker has no feasible partner in any stable matching, which is a contradiction.
\end{proof}

\begin{lemma}
\label{lem:daalgorithm_worker_both_sides2}
    If \Cref{alg:daalgorithm_worker_both_sides} returns a matching, it must be stable under both $A$ and $B$.
\end{lemma}

\begin{proof}[Proof of \Cref{lem:daalgorithm_worker_both_sides2}]
Suppose $\mu$ is the perfect matching returned, and without loss of generality, suppose it has a blocking pair $(w, f)$ with respect to $A$. That is, $f >_w^A \mu(w)$ and $w >_f^A \mu(f)$. Then $w$ must have proposed to $f$ during some iteration of the algorithm and been rejected in room $A$ for some worker $w'$ whom $f$ preferred to $w$, eventually ending up with $\mu(f)$. As firms only get better matches as the algorithm progresses, it must be the case that $\mu(f) >_f w$, and hence $(w, f)$ is not a blocking pair.
\end{proof}

\begin{lemma}
\label{lem:daalgorithm_worker_both_sides3}
    If \Cref{alg:daalgorithm_worker_both_sides} returns a matching, it must be the worker-optimal matching.
\end{lemma}

\begin{proof}[Proof of \Cref{lem:daalgorithm_worker_both_sides3}]
If not, some worker $w \neq w_1$ has a better partner in the worker-optimal matching. But then, that worker would have proposed to that partner in every room and been rejected in at least one of them, resulting in a contradiction.
\end{proof}
\subsection{Proof of Correctness of Algorithm~\ref{alg:daalgorithm_firm_both_sides}}
\label{sec:daalgorithm_firm_both_sides}

The algorithm involves running the Deferred Acceptance algorithm in different rooms with firms proposing, and returning the firm-optimal matching in $\mathcal{M}_A \cap \mathcal{M}_B$, given that $A$ and $B$ are $(1,n)$. Let $w_1$ be the worker with different preferences in $A$ and $B$.

\begin{lemma}
\label{lem:daalgorithm_firm_both_sides}
    The \Cref{alg:daalgorithm_firm_both_sides} terminates and always returns a perfect matching or $\boxtimes$.
\end{lemma}

\begin{proof}[Proof of \Cref{lem:daalgorithm_firm_both_sides}]
    The algorithm terminates since there is at least one rejection per iteration. 

    Suppose the algorithm terminates with distinct perfect matchings $\mu_A$ in room $A$ and $\mu_B$ in room $B$. Then there is some worker $w \neq w_1$ who has the same lists in $A$ and $B$ but distinct matches $\mu_A(w)$ and $\mu_B(w)$. Suppose, without loss of generality, that $\mu_A(w) >_w^{A, B} \mu_B(w)$. Then $w$ must have received a proposal from $\mu_A(w)$ in some iteration, since to be tentatively matched, a worker must have received a proposal from that firm. Workers reject all firms that are less preferred in all rooms, so $w$ must have rejected $\mu_B(w)$ in room $R_A$ as well as in room $R_B$. Therefore, $\mu_B$ cannot be the final perfect matching in room $R_B$—a contradiction.
\end{proof}

\begin{lemma}
\label{lem:n_1_firm_bp}
    If firm $f$ is rejected by worker $w$ in either room, then $w$ and $f$ cannot be matched in any strongly stable matching in $\mathcal{M}_A \cap \mathcal{M}_B$.
\end{lemma}

\begin{proof}[Proof of \Cref{lem:n_1_firm_bp}]
We use the property of the DA algorithm for a single instance $A$, which states that if a firm $f$ is rejected by worker $w$, then $(w, f)$ are not matched in any stable matching with respect to $A$~\cite{GaleS}. 

We also note the additional fact that once a worker rejects a firm, they are only matched to partners they strictly prefer over that firm from that point onward. This is equivalent to preemptively rejecting all firms ranked below $f$—an idea we carry over to this algorithm.

Suppose, without loss of generality, that $w$ rejects $f$ in room $A$. Then $(w, f)$ are not matched in any stable matching in $\mathcal{M}_A$, and hence they cannot be matched in any stable matching in $\mathcal{M}_A \cap \mathcal{M}_B$.
\end{proof}
\begin{lemma}
\label{lem:daalgorithm_firm_both_sides2}
If $A$ and $B$ admit a stable matching under both, then \Cref{alg:daalgorithm_firm_both_sides} finds one.
\end{lemma}

\begin{proof}[Proof of \Cref{lem:daalgorithm_firm_both_sides2}]
Suppose there exists such a stable matching, but the algorithm terminates by outputting $\boxtimes$. This only happens when a firm is rejected by all workers. However, by \Cref{lem:n_1_firm_bp}, this means that the firm has no feasible partner in any of the stable matchings, which is a contradiction.
\end{proof}

\begin{lemma}
\label{lem:daalgorithm_firm_both_sides3}
If \Cref{alg:daalgorithm_firm_both_sides} returns a matching, it must be stable under both $A$ and $B$.
\end{lemma}

\begin{proof}[Proof of \Cref{lem:daalgorithm_firm_both_sides3}]
Suppose the algorithm returns a perfect matching $\mu$, and assume, without loss of generality, that it has a blocking pair $(w, f)$ with respect to $A$. That is, $w >_f^A \mu(f)$ and $f >_w^A \mu(w)$. Then $f$ must have proposed to $w$ during some iteration of the algorithm in room $R_A$ and been rejected for some firm $f'$ that $w$ preferred (under $A$). Eventually, $w$ ends up with $\mu(w)$, which it prefers the most. Hence, $\mu(w) >_w^A f$, contradicting the assumption that $(w, f)$ is a blocking pair.
\end{proof}

\begin{lemma}
\label{lem:daalgorithm_firm_both_sides4}
If \Cref{alg:daalgorithm_firm_both_sides} returns a matching, it must be the firm-optimal matching.
\end{lemma}

\begin{proof}[Proof of \Cref{lem:daalgorithm_firm_both_sides4}]
If not, some firm $f$ has a better partner in the firm-optimal matching. But then, that firm would have proposed to that partner in all rooms and been rejected in at least one room, leading to a contradiction.
\end{proof}

\subsection{Proof of Correctness of Algorithm~\ref{alg:robust_xp}}
\label{sec:robust_xp_proofs}

Let $A$ and $B$ be two instances on workers $W$ and firms $F$. Let $S$ be the set of agents whose preferences differ between $A$ and $B$. In a fixed iteration of \Cref{alg:robust_xp}, let $M$ denote the partial matching on $S$ chosen in Step~1(a), let $T \coloneqq S \cup M(S)$, and let $U \coloneqq (W \cup F)\setminus T$. The truncated instance $X$ on $U$ is obtained in Steps~1(c)--(d). Since all agents with differing preferences are removed, the preferences of agents in $U$ are identical in both $A$ and $B$.

\begin{lemma}
\label{lem:robust_completeness}
If there exists a robust stable matching $M^\star \in M_A \cap M_B$, then \Cref{alg:robust_xp} returns a perfect matching.
\end{lemma}

\begin{proof}[Proof of \Cref{lem:robust_completeness}]
Consider the iteration that selects the partial matching $M \coloneqq M^\star|_S$ in Step~1(a). Let $M' \coloneqq M^\star|_U$. 

First, note that there are no blocking pairs $(w, f)$ with $w, f \in T$ for $M$ under either instance $A$ or $B$; otherwise, $M^\star$ would not be stable in the corresponding instance. Therefore, the algorithm does not abort at Step~1(b) during this iteration.

Next, observe that $M'$ is also stable in $X$. Otherwise, any blocking pair for $M'$ in $X$ would correspond to a blocking pair in both $A$ and $B$, contradicting the stability of $M^\star$. Furthermore, since $M^\star$ is a perfect matching on $W \cup F$, $M'$ must be a perfect matching on $U$.

Since $M'$ is a perfect stable matching in $X$, the Rural Hospital Theorem implies that every stable matching in $X$ is perfect. Hence, the algorithm must output a perfect matching in this iteration.

\end{proof}

\begin{lemma}
\label{lem:robust_soundness}
If \Cref{alg:robust_xp} returns a perfect matching $M^\star$, then $M^\star \in \mathcal{M}_A \cap \mathcal{M}_B$.
\end{lemma}

\begin{proof}[Proof of \Cref{lem:robust_soundness}]
Let $M \coloneqq M^\star|_S$ be the partial matching on $S$, and let $M' \coloneqq M^\star|_U$ be the perfect stable matching in $X$ obtained in the iteration in which the algorithm returned $M^\star$. Fix any instance $I \in \{A, B\}$. Any potential blocking pair $(a, b)$ under $I$ must fall into one of the following cases:
\begin{enumerate}
    \item \emph{Both agents in $T$.} This is ruled out by Step~1(b), which discards the iteration if any such pair blocks $M$ in either instance.
    
    \item \emph{One agent in $T$ and the other in $U$.} Without loss of generality, let $a \in T$ and $b \in U$. If $(a, b)$ is a blocking pair under $I$, then $b \succ_a^I M(a)$ and $a \succ_b^I M'(b)$. However, since $b \succ_a^I M(a)$, agent $a$ and all agents ranked worse than $a$ are removed from $b$'s preference list in instance $X$ during Step~1(d). Thus, $b$ cannot be matched to an agent ranked below $a$, contradicting $a \succ_b^I M'(b)$.
    
    \item \emph{Both agents in $U$.} If $(a, b)$ is a blocking pair, then $b \succ_a^I M'(a)$ and $a \succ_b^I M'(b)$. Note that we only truncate the lower ends of each agent’s preference list. In particular, if $x \in U$ is matched to $M'(x)$, then every agent that $x$ prefers to $M'(x)$ in instance $I$ must still appear in $x$'s preference list in instance $X$. Therefore, the pair $(a, b)$ must also be blocking under instance $X$, contradicting the fact that Step~1(e) returned $M'$ as stable.
\end{enumerate}
Thus $M^\star$ admits no blocking pair in $I$, and since $I \in \{A, B\}$ was arbitrary, $M^\star \in \mathcal{M}_A \cap \mathcal{M}_B$.
\end{proof}

\begin{lemma}
\label{lem:robust_xp}
Algorithm~\ref{alg:robust_xp} decides the existence of a robust stable matching in time $O\!\left(n^{p+q+2}\right)$.
\end{lemma}

\begin{proof}[Proof of \Cref{lem:robust_xp}]
We bound the work per iteration and the total number of iterations:
\begin{enumerate}
    \item \emph{Outer enumeration.} Step~1 iterates over all injective partner assignments to the $|S| = p + q$ agents, which yields $O(n^{p+q})$ possibilities.
    
    \item \emph{Per-iteration work.} Step~1(b) checks for blocking pairs within $T$; Steps~1(c)--(d) construct the truncated instance $X$; Step~1(f) finds a stable matching in $X$; and Step~1(g) verifies whether $M'$ is a perfect matching on $U$. Each of these steps can be implemented in $O(n^2)$ time using standard data structures from the $O(n^2)$ implementation of the Deferred Acceptance algorithm.
\end{enumerate}

Hence, the overall runtime is $O(n^{p+q+2})$. Correctness follows from \Cref{lem:robust_completeness,lem:robust_soundness}.
\end{proof}

\begin{lemma}
\label{lem:robust_xp_enumeration}
The set of robust stable matchings of instances $A$ and $B$ can be enumerated with $O(n^{p+q+2})$ delay.
\end{lemma}

\begin{proof}[Proof of \Cref{lem:robust_xp_enumeration}]
\Cref{alg:robust_xp} can be modified to enumerate all robust stable matchings as follows: instead of returning a single stable matching in Step~1(f), enumerate all stable matchings of the truncated instance $X$ using a known enumeration algorithms of its stable matching lattice structure with $O(n^2)$ delay per matching.

Since the outer loop of Step~1 iterates over $O(n^{p+q})$ choices of partial matchings on $S$, and each such iteration performs $O(n^2)$ work (or terminates early at Step~1(b)), the delay between any two outputs is at most $O(n^{p+q+2})$.

Hence, all robust stable matchings can be enumerated with $O(n^{p+q+2})$ delay.
\end{proof}

Together, the preceding lemmas establish \Cref{thm:robust_xp}.

\robustalgoworks*

\section{Some Surprising Examples}
\label{app:examples}

In this section, we assume that agents on both sides change preferences arbitrarily. We will demonstrate examples showing that the \emph{join} and \emph{meet} operations are no longer equivalent in $\Lc_A$ and $\Lc_B$. As a result, the argument used to prove \Cref{prop.sublattice_one_side} does not extend. In fact, for stable matchings $M, M' \in \mathcal{M}_A \cap \mathcal{M}_B$, it may be the case that some workers prefer $M$ to $M'$ in $A$ but prefer $M'$ to $M$ in $B$, which is not possible in Setting I. Specifically, even if the intersection is a sublattice, the base set is $\mathcal{M}_A \cap \mathcal{M}_B$, but the partial orders $\leq_A$ and $\leq_B$ may differ.

A trial-and-error approach shows that the intersection may form a sublattice in many instances. The structure of the sublattice may be more general, as the \emph{join} and \emph{meet} operators may not be equivalent in $\Lc_A$ and $\Lc_B$ (i.e., $(\vee_A, \wedge_A) \neq (\vee_B, \wedge_B)$). The example in \Cref{ex:both_sides_twisted_lattice} demonstrates this: the intersection is a sublattice, but the partial ordering of the stable matchings in $\mathcal{M}_A \cap \mathcal{M}_B$ differs.

\begin{figure}[ht]
	\begin{wbox}
	\centering
	\begin{minipage}{.45\linewidth}
        \centering
		\begin{tabular}{p{0.5cm}|p{0.5cm}p{0.5cm}p{0.5cm}p{0.5cm}p{0.5cm}p{0.5cm}}
			1  & a & b & c & d & e & f\\
			2  & b & a & c & d & e & f\\
			3  & c & d & a & b & e & f\\
			4  & d & c & a & b & e & f\\
			5  & e & f & a & b & c & d\\
			6  & f & e & a & b & c & d
		\end{tabular}
		
		\hspace{0.2cm}
		
		Worker preferences in $A$
	\end{minipage}%
	\begin{minipage}{.45\linewidth}
        \centering
		\begin{tabular}{p{0.5cm}|p{0.5cm}p{0.5cm}p{0.5cm}p{0.5cm}p{0.5cm}p{0.5cm}}
            a  & 2 & 1 & 3 & 4 & 5 & 6\\
			b  & 1 & 2 & 3 & 4 & 5 & 6\\
			c  & 4 & 3 & 1 & 2 & 5 & 6\\
			d  & 3 & 4 & 1 & 2 & 5 & 6\\
			e  & 6 & 5 & 1 & 2 & 3 & 4\\
			f  & 5 & 6 & 1 & 2 & 3 & 4
		\end{tabular}
		
		\hspace{0.2cm}
		
		Firm preferences in $A$
	\end{minipage}

	\hspace{0.5cm}

	\begin{minipage}{.45\linewidth}
        \centering
		\begin{tabular}{p{0.5cm}|p{0.5cm}p{0.5cm}p{0.5cm}p{0.5cm}p{0.5cm}p{0.5cm}}
            \cellcolor{red!25}1  & \cellcolor{red!25}b & \cellcolor{red!25}a & c & d & e & f\\
			\cellcolor{red!25}2  & \cellcolor{red!25}a & \cellcolor{red!25}b & c & d & e & f\\
			3  & c & d & a & b & e & f\\
			4  & d & c & a & b & e & f\\
			5  & e & f & a & b & c & d\\
			6  & f & e & a & b & c & d
		\end{tabular}
		
		\hspace{0.2cm}
		
		Worker preferences in $B$
	\end{minipage}%
	\begin{minipage}{.45\linewidth}
        \centering
		\begin{tabular}{p{0.5cm}|p{0.5cm}p{0.5cm}p{0.5cm}p{0.5cm}p{0.5cm}p{0.5cm}}
            \cellcolor{red!25}a  & \cellcolor{red!25}1 & \cellcolor{red!25}2 & 3 & 4 & 5 & 6\\
			\cellcolor{red!25}b  & \cellcolor{red!25}2 & \cellcolor{red!25}1 & 3 & 4 & 5 & 6\\
			c  & 4 & 3 & 1 & 2 & 5 & 6\\
			d  & 3 & 4 & 1 & 2 & 5 & 6\\
			e  & 6 & 5 & 1 & 2 & 3 & 4\\
			f  & 5 & 6 & 1 & 2 & 3 & 4
		\end{tabular}
		
		\hspace{0.2cm}
		
		Firm preferences in $B$
	\end{minipage}
	\end{wbox}
	\caption{Worker and firm preferences (from most preferred to least)}
	\label{ex:both_sides_twisted_lattice}
\end{figure}

Instance $B$ is obtained from instance $A$ by permuting the preferences of workers $1$ and $2$, and firms $a$ and $b$. It can be seen that taking the join (or meet) of any two matchings in $\mathcal{M}_A \cap \mathcal{M}_B$ yields another matching in the intersection, regardless of which instance is considered. However, workers $1$ and $2$ prefer $M_2$ to $M_1$ in $A$, while they prefer $M_1$ to $M_2$ in $B$. \Cref{fig.example} illustrates this change in the partial ordering of matchings in $\mathcal{M}_A \cap \mathcal{M}_B$, depending on the underlying instance.

$M_1 = \{1b, 2a, 3d, 4c, 5e, 6f\}$ and $M_2 = \{1a, 2b, 3c, 4d, 5f, 6e\}$ are both stable under $A$ and $B$. \\
$X_1 = M_1 \wedge_A M_2 = \{1a, 2b, 3c, 4d, 5e, 6f\}$ and $X_2 = M_1 \vee_A M_2 = \{1b, 2a, 3d, 4c, 5f, 6e\}$. \\
$Y_1 = M_1 \wedge_B M_2 = \{1b, 2a, 3c, 4d, 5e, 6f\}$ and $Y_2 = M_1 \vee_B M_2 = \{1a, 2b, 3d, 4c, 5f, 6e\}$.

\begin{figure}[ht]
\centering
\begin{tikzpicture}
\tikzset{
mydot/.style={
  fill,
  circle,
  inner sep=1.5pt
  }
}
\path (0, 0) coordinate (A) (2, 2) coordinate (B) (4, 0) coordinate (C) (2, -2) coordinate (D);
\draw (A)-- (B) node [midway, above]{};
\draw (B)-- (C) node [midway, above]{};
\draw (A)-- (D) node [midway, above]{};
\draw (D)-- (C) node [midway, above]{};

\path (0, -2) coordinate (E) (2, 0) coordinate (F) (4, -2) coordinate (G) (2, -4) coordinate (H);
\draw (E)-- (F) node [midway, above]{};
\draw (F)-- (G) node [midway, above]{};
\draw (E)-- (H) node [midway, above]{};
\draw (H)-- (G) node [midway, above]{};

\draw (A)-- (E) node [midway, above]{};
\draw (B)-- (F) node [midway, above]{};
\draw (C)-- (G) node [midway, above]{};
\draw (D)-- (H) node [midway, above]{};

\node[mydot,label={[align=center]left:$\color{blue}Y_1$}] at (A) {};
\node[mydot,label={[align=center]above:$\cellcolor{red!25}X_1$}] at (B) {};
\node[mydot,label={[align=center]right:$M_2$}] at (C) {};
\node[mydot,label={[align=center]right:$ $}] at (D) {};
\node[mydot,label={[align=center]left:$M_1$}] at (E) {};
\node[mydot,label={[align=center]right:$ $}] at (F) {};
\node[mydot,label={[align=center]right:$\color{blue}Y_2$}] at (G) {};
\node[mydot,label={[align=center]below:$\cellcolor{red!25}X_2$}] at (H) {};

\end{tikzpicture}
\hspace{2cm}
\begin{tikzpicture}
\tikzset{
mydot/.style={
  fill,
  circle,
  inner sep=1.5pt
  }
}
\path (0, 0) coordinate (A) (2, 2) coordinate (B) (4, 0) coordinate (C) (2, -2) coordinate (D);
\draw (A)-- (B) node [midway, above]{};
\draw (B)-- (C) node [midway, above]{};
\draw (A)-- (D) node [midway, above]{};
\draw (D)-- (C) node [midway, above]{};

\path (0, -2) coordinate (E) (2,0) coordinate (F) (4, -2) coordinate (G) (2, -4) coordinate (H);
\draw (E)-- (F) node [midway, above]{};
\draw (F)-- (G) node [midway, above]{};
\draw (E)-- (H) node [midway, above]{};
\draw (H)-- (G) node [midway, above]{};

\draw (A)-- (E) node [midway, above]{};
\draw (B)-- (F) node [midway, above]{};
\draw (C)-- (G) node [midway, above]{};
\draw (D)-- (H) node [midway, above]{};

\node[mydot,label={[align=center]left:$M_1$}] at (A) {};
\node[mydot,label={[align=center]above:$\color{blue}Y_1$}] at (B) {};
\node[mydot,label={[align=center]right:\cellcolor{red!25}$X_1$}] at (C) {};
\node[mydot,label={[align=center]right:$ $}] at (D) {};
\node[mydot,label={[align=center]left:\cellcolor{red!25}$X_2$}] at (E) {};
\node[mydot,label={[align=center]right:$ $}] at (F) {};
\node[mydot,label={[align=center]right:$M_2$}] at (G) {};
\node[mydot,label={[align=center]below:$\color{blue}Y_2$}] at (H) {};

\end{tikzpicture}
\caption{Lattices $\Lc_A$ and $\Lc_B$ for example in \Cref{ex:both_sides_twisted_lattice}}
\label{fig.example}
\end{figure}

\section{Algorithms when One Side Changes Preferences}
\label{app:one_side}

In this setting, we consider instances where the workers' preferences remain unchanged while the firms' preferences differ, i.e., instances that are $(0,n)$. The case of $(n,0)$ is equivalent by symmetry. We begin by recalling the following proposition from \cite{GMRV-nearby-instances}.

\begin{proposition}
\label{prop.sublattice_one_side} (\cite{GMRV-nearby-instances}, Proposition 6)
If $A$ and $B$ are $(0,n)$, then the matchings in $\mathcal{M}_A \cap \mathcal{M}_B$ form a sublattice in each of the two lattices.
\end{proposition}

\begin{proof}[Proof of \Cref{prop.sublattice_one_side}]
If $| \mathcal{M}_A \cap \mathcal{M}_B| \leq 1$, then $\mathcal{M}_A \cap \mathcal{M}_B$ is trivially a sublattice of both $\Lc_A$ and $\Lc_B$. So assume $| \mathcal{M}_A \cap \mathcal{M}_B| > 1$ and let $M_1$ and $M_2$ be matchings in $\mathcal{M}_A \cap \mathcal{M}_B$. Let $\vee_A$ and $\vee_B$ denote the join operations under $A$ and $B$, respectively, and let $\wedge_A$ and $\wedge_B$ denote the meet operations.

By the definition of the join operation in \Cref{subsection.latticeOfSM}, $M_1 \vee_A M_2$ is the matching where each worker is assigned to their less-preferred partner (or equivalently, each firm to its more-preferred partner) among $M_1$ and $M_2$ under instance $A$. Since every worker has the same preference list in both $A$ and $B$, the less-preferred partner of each worker between $M_1$ and $M_2$ is the same under both instances. Therefore, $M_1 \vee_A M_2 = M_1 \vee_B M_2$. A similar argument holds for the meet, showing $M_1 \wedge_A M_2 = M_1 \wedge_B M_2$. Thus, both $M_1 \vee_A M_2$ and $M_1 \wedge_A M_2$ are in $\mathcal{M}_A \cap \mathcal{M}_B$, proving that $\mathcal{M}_A \cap \mathcal{M}_B$ forms a sublattice in both lattices.
\end{proof}

The proposition shows that the join and meet of stable matchings in $\mathcal{M}_A \cap \mathcal{M}_B$ are the same under the ordering of $\Lc_A$ and $\Lc_B$. This implies that workers agree on which matchings in the intersection they prefer. The key driver of this result is that workers have the \emph{same preferences} in both instances, allowing many arguments used in single-instance stable matching lattices to carry over. Additionally, the set $\mathcal{M}_A \cap \mathcal{M}_B$ forms the same sublattice in both $\Lc_A$ and $\Lc_B$ in the sense that the partial orders coincide; that is, $(\mathcal{M}_A \cap \mathcal{M}_B, \preceq_A) = (\mathcal{M}_A \cap \mathcal{M}_B, \preceq_B)$. This implies the existence of worker-optimal and firm-optimal stable matchings in $\mathcal{M}_A \cap \mathcal{M}_B$, and suggests that a variant of the Deferred Acceptance algorithm could succeed in this setting.

This motivates the notion of a *compound instance* that encodes agent preferences across both input instances. We present a variant of the Deferred Acceptance algorithm (\Cref{alg:daalgorithm_firm_one_side}) on this new instance to find such matchings. See \Cref{fig:compoundinstance} for the construction.

\begin{definition}
\label{def:compound_instance}
Let $A$ and $B$ be $(0,n)$. Define the compound instance $X$ as follows:
\begin{enumerate}
    \item For each worker $w \in \mathcal{W}$, the preference list in $X$ is the same as in $A$ and $B$.
    \item For each firm $f \in \mathcal{F}$, $f$ prefers worker $w_i$ to $w_j$ in $X$ if and only if $f$ prefers $w_i$ to $w_j$ in both $A$ and $B$:
    \[
    \forall f \in \mathcal{F}, \forall w_i, w_j \in \mathcal{W}, \quad w_i >^X_f w_j \iff (w_i >^A_f w_j \text{ and } w_i >^B_f w_j).
    \]
\end{enumerate}
\end{definition}

In the compound instance $X$, workers have totally ordered preference lists, while firms may have partial orders. If a firm $f$ does not strictly prefer $w_i$ to $w_j$ in $X$, we say $f$ is indifferent between them, denoted $w_i \sim^X_f w_j$.

Given a matching $M$, a worker-firm pair $(w, f) \not\in M$ is said to be a strongly blocking pair with respect to $M$ if $w$ prefers $f$ to their assigned partner in $M$ and $f$ either strictly prefers $w$ or is indifferent between $w$ and its partner in $M$. A matching with no strongly blocking pairs is said to be \emph{strongly stable} under $X$. Let $\mathcal{M}_X$ denote the set of matchings that are strongly stable under $X$. See \cite{IRVING-Indifferences,Fleiner} for more information on Strongly stable matchings.

The following lemma shows that compound instances are meaningful in our context:

\begin{lemma} \cite{Fleiner}
\label{lem:strongly_stable}
If $X$ is the compound instance of $A$ and $B$, then a matching $M$ is strongly stable under $X$ if and only if it is stable under both $A$ and $B$. That is, $\mathcal{M}_X = \mathcal{M}_A \cap \mathcal{M}_B$.
\end{lemma}

\begin{proof}[Proof of \Cref{lem:strongly_stable}]
Let $(w, f)$ be a strongly blocking pair for matching $M$ under $X$. Since $w$ has identical preferences in $A$ and $B$, they must strictly prefer $f$ to $M(w)$ in both instances. Firm $f$ must either strictly prefer $w$ to its partner in $M$ or be indifferent between them. In either case, $f$ strictly prefers $w$ in at least one of $A$ or $B$, making $(w, f)$ a blocking pair under that instance. Thus, $M$ is not stable under both.

Conversely, if $(w, f)$ is a blocking pair under $A$ or $B$, then it satisfies the conditions for being strongly blocking under $X$. Therefore, $\mathcal{M}_X = \mathcal{M}_A \cap \mathcal{M}_B$.
\end{proof}
\begin{algorithm}[ht]
	\begin{wbox}
		\textsc{DeferredAcceptanceAlgorithm}($A$): \\
		\textbf{Input:} Stable matching instance $A$ \\
		\textbf{Output:} Perfect matching $M$ \\
        Each worker has a list of all firms in decreasing order of priority.
		\begin{enumerate}
			\item Until all firms receive a proposal, do:
			\begin{enumerate}
				\item $\forall w \in W$, $w$ proposes to their best uncrossed firm.
				\item $\forall f \in F$, $f$ tentatively accepts their \textit{\textbf{best}} proposal and rejects the rest.
				\item $\forall w \in W$, if $w$ is rejected by a firm $f$, they cross $f$ off their list.
			\end{enumerate}  
			\item Output the perfect matching $M$.
		\end{enumerate}
	\end{wbox}
	\caption{Deferred Acceptance Algorithm with workers proposing, taken from \cite{Book-Online}}
	\label{alg:gs_algorithm} 
\end{algorithm} 

\begin{algorithm}[ht]
	\begin{wbox}
		\textsc{CompoundInstance}($A, B$): \\
		\textbf{Input:} Stable matching instances $A$ and $B$ with the same preferences on the workers' side. \\
		\textbf{Output:} Instance $X$.
		\begin{enumerate}
			\item $\forall w \in W$, $w$'s preferences in $X$ are the same as in $A$ and $B$.
			\item $\forall f \in F$, $\forall w_i, w_j \in W$, $w_i >_f^X w_j$ if and only if $w_i >_f^A w_j$ and $w_i >_f^B w_j$.
			\item Return $X$.
		\end{enumerate}
	\end{wbox}
	\caption{Subroutine for constructing a compound instance.}
	\label{fig:compoundinstance} 
\end{algorithm} 

\subsection{Worker and firm optimal matchings}
\label{subsection.alg_one_sided_errors}

\begin{algorithm}[ht]
	\begin{wbox}
		\textsc{WorkerOptimalForCompoundInstance}($X$): \\
		\textbf{Input:} Compound stable matching instance $X$ \\
		\textbf{Output:} Perfect matching $M$ or $\boxtimes$  \\
        Each worker has a list of all firms, initially uncrossed, in decreasing order of priority.
		\begin{enumerate}
			\item Until all workers receive an acceptance or some worker is rejected by all firms, do:
			\begin{enumerate}
				\item $\forall w \in W$: $w$ proposes to their best uncrossed firm.
				\item $\forall f \in F$: $f$ tentatively accepts their \textbf{\textit{best-ever}} proposal, i.e., the proposal that they strictly prefer to every other proposal ever received, and rejects the rest.
				\item $\forall w \in W$: if $w$ is rejected by a firm $f$, they cross $f$ off their list. 
			\end{enumerate}  
			\item If all workers get an acceptance, output the perfect matching ($M$); else output $\boxtimes$.
		\end{enumerate}
	\end{wbox}
	\caption{\cite{IRVING-Indifferences}Finding the worker-optimal stable matching using compound instance.}
	\label{alg:daalgorithm_worker_one_side} 
\end{algorithm} 

\begin{algorithm}[ht]
	\begin{wbox}
		\textsc{FirmOptimalForCompoundInstance}($X$): \\
		\textbf{Input:} Compound stable matching instance $X$ \\
		\textbf{Output:} Perfect matching $M$ or $\boxtimes$ \\
        Each firm has a DAG over all workers, initially uncrossed, representing its partial order preferences.
		\begin{enumerate}
			\item Until all firms receive an acceptance or some firm is rejected by all workers, do:
			\begin{enumerate}
				\item $\forall f \in F$: $f$ proposes to the \textit{\textbf{best uncrossed workers}}, i.e., all uncrossed workers $w$ such that there is no uncrossed $w'$ with $w' >^X_f w$.
				\item $\forall w \in W$: $w$ tentatively accepts the best proposal and rejects the rest.
				\item $\forall f \in F$: if $f$ is rejected by a worker $w$, they cross $w$ off their DAG.
			\end{enumerate}
			\item If all firms receive an acceptance, output the perfect matching ($M$); else output $\boxtimes$.
		\end{enumerate}
	\end{wbox}
	\caption{\cite{IRVING-Indifferences}Finding the firm-optimal stable matching using compound instance.}
	\label{alg:daalgorithm_firm_one_side} 
\end{algorithm} 

The original Deferred Acceptance Algorithm(\Cref{alg:gs_algorithm}) works in iterations. In each iteration, three steps occur: (a) the proposing side (say, workers) proposes to their best firm that hasn’t yet rejected them; (b) each firm tentatively accepts their best proposal of the round and rejects all others; (c) all workers cross off the firms that rejected them. The loop continues until a perfect matching is found and output. The key idea is that once a rejection occurs, that worker–firm pair cannot be part of any stable matching. 

We modify this algorithm to find the worker- and firm-optimal stable matchings in the intersection for the $(0, n)$ setting, ensuring that strong blocking pairs are prevented.

\begin{theorem}
\label{thm:n_0_algo_works}
\Cref{alg:daalgorithm_worker_one_side} and \Cref{alg:daalgorithm_firm_one_side} find the worker- and firm-optimal matchings respectively, if they exist.
\end{theorem}

The worker-optimal Algorithm in \Cref{alg:daalgorithm_worker_one_side} (simplified algorithm of \cite{IRVING-Indifferences} and equivalent to algorithms in \cite{NP-Two-stable, genPreferences, GMRV-nearby-instances}) differs from the classic DA algorithm as follows: firms accept a proposal only if it is strictly better than every proposal they have ever received (including from previous rounds). This avoids strong blocking pairs involving firm indifference. Note that a firm may reject all proposals in a round if it is indifferent among the proposing workers, which is necessary to maintain strong stability.

In the firm-optimal Algorithm in \Cref{alg:daalgorithm_firm_one_side} (equivalent to Algorithm SUPER in \cite{IRVING-Indifferences}), firms propose using a DAG of partial preferences. Each firm proposes to all uncrossed workers $w$ such that no uncrossed worker $w'$ exists with $w' >^X_f w$. Workers (with total preferences) tentatively accept their best proposal and reject the rest. If a worker rejects a firm, then by strict preferences, they cannot form a blocking pair later. If multiple workers tentatively accept a firm’s proposals, some may later reject the firm, and the algorithm either continues or terminates with no perfect matching.

Note: when $A = B$, then $X$ has total preferences on both sides, and these algorithms reduce to the standard DA algorithm.

\begin{remark} (\cite{Fleiner})
Strongly stable matchings in an instance $Y$ with partial order preferences on one side form a lattice. Since each partial order has finitely many linear extensions, $\mathcal{M}_Y$ can be written as the intersection of finitely many lattices. However, as the number of linear extensions may be exponential in $n$, this does not yield an efficient method for computing the Birkhoff partial order.
\end{remark}

\subsection{Proof of Correctness of Algorithm~\ref{alg:daalgorithm_worker_one_side}}
\label{sec:daalgorithm_worker_one_side}

We have instances $A$ and $B$ which are $(0,n)$, meaning that workers have complete order preferences and firms have partial order preferences in $X$. In \Cref{alg:daalgorithm_worker_one_side}, workers propose, and unlike in \Cref{alg:gs_algorithm}, firms tentatively accept only the worker whom they strictly prefer to every other proposer they have received—in the current round and all previous rounds.

\begin{lemma}
\label{lem:n_0_worker_bp}
If firm $f$ rejects worker $w$ in any round of \Cref{alg:daalgorithm_worker_one_side}, then $w$ and $f$ can never be matched in any matching that is strongly stable under $X$.
\end{lemma}

\begin{proof}[Proof of \Cref{lem:n_0_worker_bp}]
Assume there is a rejection (otherwise, the result is trivial). Let firm $f$ reject worker $w$. We will show that any matching $M$ such that $(w,f) \in M$ has a strong blocking pair. We proceed by induction on the round number in which the rejection occurs.

Assume the rejection happens in the first round. Then $f$ only rejects $w$ if it receives a proposal from another worker $w'$ whom it prefers or is indifferent to. Since $w'$ proposed to $f$ in the first round, $w'$ must prefer $f$ to $M(w')$. Hence, $(w', f)$ is a strong blocking pair to $M$.

Now assume the rejection occurs in round $k > 1$. If $f$ rejects $w$, then there exists a worker $w'$ who has proposed to $f$ in this or a previous round and such that $w' >^X_f w$ or $w' \sim^X_f w$. If $w'$ was rejected by $f$ in an earlier round, then by the induction hypothesis, $w'$ cannot be matched to any firm it proposed to before $f$, including $f$, in any stable matching. Hence, $f >_{w'} M(w')$, and so $(w', f)$ is a strong blocking pair.
\end{proof}

\begin{lemma}
\label{lem:n_0_worker_bp2}
If $X$ admits a strongly stable matching, then \Cref{alg:daalgorithm_worker_one_side} finds one.
\end{lemma}

\begin{proof}[Proof of \Cref{lem:n_0_worker_bp2}]
Suppose there exists a strongly stable matching, but the algorithm terminates by outputting $\boxtimes$. This only happens when a worker is rejected by all firms. However, by \Cref{lem:n_0_worker_bp}, this means that the worker has no feasible partner in any strongly stable matching, which is a contradiction.
\end{proof}

\begin{lemma}
\label{lem:n_0_worker_bp3}
If \Cref{alg:daalgorithm_worker_one_side} returns a matching, it must be strongly stable under $X$.
\end{lemma}

\begin{proof}[Proof of \Cref{lem:n_0_worker_bp3}]
Assume not. Let $(w, f)$ be a blocking pair for the matching returned by the algorithm. Then $w$ must have proposed to $f$ during the execution and been rejected. Hence, $f$ must be matched to someone it strictly prefers at the end of the algorithm, implying that $(w, f)$ is not a strong blocking pair—a contradiction.
\end{proof}

\begin{lemma}
\label{lem:n_0_worker_bp4}
If \Cref{alg:daalgorithm_worker_one_side} returns a matching, it must be the worker-optimal strongly stable matching.
\end{lemma}

\begin{proof}[Proof of \Cref{lem:n_0_worker_bp4}]
If not, then some worker $w$ has a better partner in the worker-optimal matching. But in that case, $w$ must have proposed to that partner and been rejected during the algorithm, leading to a contradiction.
\end{proof}

\subsection{Proof of Correctness of Algorithm~\ref{alg:daalgorithm_firm_one_side}}
\label{sec:daalgorithm_firm_one_side}

We again have instances $A$ and $B$ which are $(0,n)$, meaning that workers have complete order preferences and firms have partial order preferences in $X$. In \Cref{alg:daalgorithm_firm_one_side}, firms propose, and unlike in \Cref{alg:gs_algorithm}, they may propose to multiple workers in the same iteration.

\begin{lemma}
\label{lem:n_0_firm_bp}
If worker $w$ rejects firm $f$ in any round of \Cref{alg:daalgorithm_firm_one_side}, then $w$ and $f$ can never be matched in any matching stable under both instances $A$ and $B$.
\end{lemma}

\begin{proof}[Proof of \Cref{lem:n_0_firm_bp}]
Suppose, for contradiction, that $w$ rejects firm $f$ in some round, yet they are matched in a stable matching $\mu$.

We proceed by induction on the round number. Suppose $w$ rejects $f$ in the first round. Then there is another firm $f_1$ that $w$ strictly prefers to $f$. Since firms propose to their most preferred uncrossed workers, and this is the first round, there is no worker $w'$ such that $f_1$ prefers to $w$; i.e., $w >_{f_1} w'$ or $w \sim_{f_1} w'$. In matching $\mu$, where $w$ is matched to $f$, this implies that $(w, f_1)$ is a strong blocking pair.

Suppose $w$ rejects $f$ in the $k$th round. Then there is a firm $f_k$ whom $w$ strictly prefers to $f$. One of the following must hold:
\begin{enumerate}
    \item $w >_{f_k} \mu(f_k)$ or $w \sim_{f_k} \mu(f_k)$: Then $(w, f_k)$ is a strong blocking pair.
    \item $\mu(f_k) >_{f_k} w$: Then $f_k$ must have proposed to $\mu(f_k)$ in a previous round and been rejected. By the induction hypothesis, this leads to a different strong blocking pair $(w', f')$ for $\mu$, implying that $\mu$ is not stable.
\end{enumerate}
\end{proof}

\begin{lemma}
\label{lem:n_0_firm_bp2}
\Cref{alg:daalgorithm_firm_one_side} terminates.
\end{lemma}

\begin{proof}[Proof of \Cref{lem:n_0_firm_bp2}]
For any round of the algorithm, one of the following holds:
\begin{enumerate}
    \item Fewer than $n$ proposals are made in round $k$: Then some firm $f$ has no uncrossed workers left to propose to in its list (DAG). The algorithm terminates before this round and outputs $\boxtimes$.
    \item Exactly $n$ proposals are made to $n$ distinct workers: Then all workers receive a single proposal, and either a rejection occurs or the algorithm terminates at the end of the round with a perfect matching.
    \item More than $n$ proposals are made: Then at least one worker $w$ receives multiple proposals, rejects at least one, and at least one firm crosses $w$ off its list.
\end{enumerate}
Thus, in every round, there is either a rejection or the algorithm terminates. Since there are at most $O(n^2)$ rejections, the algorithm terminates in $\poly(n)$ rounds.
\end{proof}

\begin{lemma}
\label{lem:n_0_firm_bp3}
If $X$ admits a strongly stable matching, then \Cref{alg:daalgorithm_firm_one_side} finds one.
\end{lemma}

\begin{proof}[Proof of \Cref{lem:n_0_firm_bp3}]
Suppose there exists a strongly stable matching, but the algorithm terminates by outputting $\boxtimes$. This only happens when a firm is rejected by all workers. However, by \Cref{lem:n_0_firm_bp}, this implies that the firm has no feasible partner in any strongly stable matching, which is a contradiction.
\end{proof}

\begin{lemma}
\label{lem:n_0_firm_bp4}
If \Cref{alg:daalgorithm_firm_one_side} returns a matching, it must be strongly stable under $X$.
\end{lemma}

\begin{proof}[Proof of \Cref{lem:n_0_firm_bp4}]
Assume not. Let $(w, f)$ be a blocking pair to the returned matching. Then $f$ must have proposed to $w$ during the algorithm and been rejected. Thus, $w$ must be matched to a more preferred firm at the end of the algorithm, implying $(w, f)$ is not a strong blocking pair—a contradiction.
\end{proof}

\begin{lemma}
\label{lem:n_0_firm_bp5}
If \Cref{alg:daalgorithm_firm_one_side} returns a matching, it must be the firm-optimal strongly stable matching.
\end{lemma}

\begin{proof}[Proof of \Cref{lem:n_0_firm_bp5}]
If not, then some firm $f$ has a better partner in the firm-optimal matching. But then, that firm must have proposed to that partner and been rejected during the algorithm, which leads to a contradiction.
\end{proof}

Note that all the above proofs hold for any compound instance $X$ that has partial order preferences on one side and complete order preferences on the other. Thus, the setting naturally generalizes beyond just two instances, as observed in \Cref{thm:n_1_multiple}. Also note that both algorithms behave exactly like the original Deferred Acceptance algorithm when $X$ has complete order preferences on both sides; that is, when instances $A$ and $B$ are the same.

\end{document}